\documentclass{article}
\usepackage[lang = british]{ems-msl} 

\numberwithin{equation}{section}
\usepackage{comment}
\usepackage{color}
\usepackage{float}
\usepackage{bbm}
\floatstyle{boxed}
\newfloat{algorithm}{t}{lop}
\usepackage{tikz}
\usepackage{varwidth}

\usepackage{amstext,amsopn} 
\usepackage{amsmath}%
\usepackage{amsthm}%

\newtheorem{observation}{Observation}
\newtheorem{theorem}{Theorem}
\newtheorem{example}{Example}
\newtheorem{lemma}{Lemma}
\newtheorem{proposition}{Proposition}
\newtheorem{remark}{Remark}
\newtheorem{corollary}{Corollary}

\usepackage{url}

\DeclareMathOperator{\Ex}{\mathbb{E}}

\hypersetup{colorlinks=true,linkcolor=blue,citecolor=blue,anchorcolor = blue}

\renewcommand*{\Pr}{\mathbb{P}}

\begin{document}

\title{Multiplayer Bandit Learning, from Competition to Cooperation}
\titlemark{Multiplayer Bandit Learning}

\emsauthor{1}{Simina Br\^anzei}{S.~Br\^anzei}
\emsauthor{2}{Yuval Peres}{Y.~Peres}


\emsaffil{1}{Postal address: Purdue University; \email{simina.branzei@gmail.com}}
\emsaffil{2}{Postal address: Kent State University; \email{yuval@yuvalperes.com}}

\classification[60Gxx]{91Axx}


\begin{abstract}
The stochastic multi-armed bandit model captures the tradeoff between exploration and exploitation. We study the effects of competition and cooperation on this tradeoff. Suppose there are two arms, one predictable and one risky, and two players, Alice and Bob. In every round, each player pulls an arm, receives the resulting reward, and observes the choice of the other player but not their reward. Alice's utility is $\Gamma_A + \lambda \Gamma_B$ (and similarly for Bob), where $\Gamma_A$ is Alice's total reward and $\lambda \in [-1, 1]$ is a cooperation parameter. At $\lambda = -1$ the players are competing in a zero-sum game, at $\lambda = 1$, their interests are aligned, and at $\lambda = 0$, they are neutral: each player's utility is their own reward. The model is related to the economics literature on strategic experimentation, where usually players observe each other's rewards.

Suppose the predictable arm has success probability $p$ and the risky arm has prior $\mu$. If the discount factor is $\beta$, then the value of $p$ where a single player is indifferent between the arms is the Gittins index $g = g(\mu,\beta) > m$, where $m$ is the mean of the risky arm.

We answer, in this setting, a fundamental question posed by \cite{rotschild}. We show that competing and neutral players eventually settle on the same arm (even though it may not be the best arm) in every Nash equilibrium, while this can fail for players with aligned interests.

Moreover, we show that \emph{competing players} explore \emph{less} than a single player: there is $p^* \in (m, g)$ so that for all $p > p^*$, the players stay at the predictable arm. However, the players are not myopic: they still explore for some $p > m$. On the other hand, \emph{cooperating players} (with $\lambda =1$) explore \emph{more} than a single player. We also show that \emph{neutral players} learn from each other, receiving strictly higher total rewards than they would playing alone, for all $ p\in (p^*, g)$, where $p^*$ is the threshold above which competing players do not explore.
\end{abstract}

\maketitle

\section{Introduction}

The multi-armed bandit learning problem is a paradigm that captures the tradeoffs between exploration and exploitation ~\cite{Gittins_book,BCB12,Slivkins_book,Tor_book}.
We study the effects of competition and cooperation on exploration in a multiplayer stochastic bandit problem, where multiple players are selecting arms and receiving the corresponding rewards in each round.
Examples of scenarios that can be captured by this model include
competing firms in a saturated market, which can be seen as a zero-sum game where the utility of a player is the difference between their rewards and the rewards of the opponent. At the other extreme, the interaction of organisms that are (almost) genetically identical, such as ants and bees, can be modeled by a fully cooperative game~\cite{Hamilton64} where the utility of a player is the sum of the rewards of all players.

We consider a unifying framework that interpolates between these extremes, and focus on the so-called one-armed bandit problem, where the choices are playing a predictable arm or a risky one. The feedback received by each player is their own reward and the action taken by the other player.
If Alice's total reward is $\Gamma_A$ and Bob's total reward is $\Gamma_B$, then Alice's utility is $\Gamma_A + \lambda \Gamma_B$ and similarly for Bob, where $\lambda \in \mathbb{R}$. We study three choices for $\lambda$: the zero-sum case ($\lambda = -1$), the fully aligned case ($\lambda= 1$), and the neutral setting ($\lambda = 0$).

We study the long term behavior of the players, and show that both neutral and competing eventually settle on the same arm (even though it may not be the best arm) in every Nash equilibrium (Theorem~\ref{thm:neutral_competing_settle_same_arm}). However, this can fail for cooperating players, where there are Nash equilibria in which one of the players switches infinitely often between the arms.

Rotschild~\cite{rotschild} studied the theory of market pricing using a two-armed bandit model and asked the question of whether players eventually settle on the same arm.
Aoyagi~\cite{aoyagi,aoyagi_corr} answered this question positively in the model with imperfect monitoring, where players can see each other's actions but only their own rewards, under an assumption on the distributions.
To the best of our knowledge, our paper is the first progress on the Rothschild conjecture since Aoyagi's work.

A key question in the zero sum scenario is to quantify the value of information obtained by exploring (i.e. experimenting with the risky arm). Note that while Alice does not know Bob's rewards, she may try to infer them from his actions, which in turn may lead to Bob trying to change his actions to hide information from her. 
We find that under optimal play, information is less valuable in the zero-sum game than in the one player setting, which leads to reduced exploration compared to the one player optimum (Theorem~\ref{thm:competing_explore_less_discounted}). However, Theorem~\ref{thm:non_myopic_discounted} shows that information still has positive value. Our model is close to the games of incomplete information analyzed by Aumann and Maschler~\cite{aumann_book}, where a player may forego some rewards to hide information from their opponent.

In contrast to the zero-sum scenario, we show that exploration is increased in the fully cooperative setting compared to the single player optimum (Theorem~\ref{thm:cooperation}). We also study the neutral regime and find that in a range of parameters, with probability $1$ the players explore in every Nash equilibrium. Moreover, if the equilibrium is also perfect Bayesian, then the players learn from each other: each player gets in expectation strictly higher total reward than they would when playing alone (Theorem~\ref{neutral}). A corollary is that with positive probability the players do not follow the same trajectory in any perfect Bayesian equilibrium.

\section{Model} \label{sec:model}

Suppose there are two players, Alice and Bob, each of which pulls one of $K$ arms in each round. The rewards are drawn from $\{0,1\}$: for each arm $k$ there is a prior distribution $\mu_k$ so that the success probability is picked (by Nature) from $\mu_k$ before the game starts. In the finite horizon version, the players play for $T+1$ rounds indexed from $0$ to $T$. In the discounted version, the game continues forever, but the players have value reduced by a factor of $\beta^t$ for the reward in round $t$, where $\beta \in (0, 1)$\footnote{This can alternatively be interpreted as the game stopping with probability $1-\beta$ in each round.}.

\paragraph{Rewards} In every round, each player collects the reward from the arm they pulled (even if both chose the same arm). We denote by $\gamma_A(t)$ and $\gamma_B(t)$ the random variable corresponding to the reward received by Alice and Bob, respectively in round $t \geq 0$. The total reward of player $i$ in a finite horizon game is $\Gamma_i = \sum_{t=0}^{T} \gamma_i(t)$, while player $i$'s reward in the discounted game is $\Gamma_i = \sum_{t = 0}^{\infty} \gamma_i(t) \cdot \beta^t$.
In every round, after selecting the arm to pull, each player observes their own reward and the action taken by the other player, but not their reward.

\paragraph{Utilities} The utility of each player is a combination of their own reward and the reward of the other player.

More precisely, there is a \emph{cooperation parameter} $\lambda \in [-1, 1]$ so that Alice's utility is $u_A = \Gamma_A + \lambda \cdot \Gamma_B$, while Bob's utility is $u_B = \Gamma_B + \lambda \cdot \Gamma_A$. For $\lambda = -1$, we have a zero-sum game while for $\lambda = 1$, the interests of the players are aligned. The case $\lambda = 0$ is the neutral regime where each player's utility is their own total reward.

\paragraph{Arms}
We focus on the so-called one-armed bandit problem, where there is a predictable left arm, denoted L, with known success probability $p$ and a risky right arm, denoted R, with a prior $\mu$ that is not a point mass. A player is said to ``explore'' if that player selects the right arm at least once.
Denote by $m =\int_{0}^{1} x \mathop{d\mu(x)}$ the mean of the prior $\mu$, by $w = \int_{0}^{1} (x-m)^2 \mathop{d\mu(x)} > 0$ the variance of $\mu$, and by $M^*$ the maximum of the support of $\mu$:
\begin{align}
	\label{eq:M_star} M^* = \sup\{x \in [0,1] : \mu(x,1] > 0\}
\end{align}

\paragraph{Strategies} A pure strategy for player $i$ is a function $S_i : \bigcup_{t \in \mathbb{N}} Y_t \times Z_t^i \rightarrow \{L, R\}$, where $Y_t$ is the public history (i.e. the sequence of past actions of both players) until the end of round $t-1$ and $Z_t^i$ is the private history of player $i$, containing the bits observed by $i$ until the end of round $t-1$.

Thus a pure strategy tells player $i$ which arm to play next given the public and private history. The space of pure strategies is a compact metrizable space in the product topology. A mixed strategy is a probability distribution over the set of pure strategies, so the space of pure strategies is weak$^*$ compact.
The \emph{expected utility} of a player is computed using the player's beliefs about the private information of the other player. 
 For more details on extensive form games, see Chapter 3 in \cite{maschler_book}.

\section{Our Results} \label{sec:our_results}

We analyze the properties of the game when the players are competing  ($\lambda = -1$), cooperating ($\lambda = 1$), and neutral ($\lambda = 0$), for both the discounted and finite horizon settings.

In the discounted setting, the \textbf{Gittins index} $g = g(\mu, \beta)$ of the right arm is defined as the infimum of the success probabilities $p$ where playing always left is optimal for a single player. The index was first defined by Bradt, Johnson, and Karlin~\cite{bradt_johnson_karlin} and its importance for the multi-armed bandit problem was shown by Gittins and Jones~\cite{gittins_jones}; see also \cite{Gittins79,berry_book}.


\subsection{Competing Players ($\lambda = -1$)}

Our first theorem shows that when using optimal strategies, competing players explore less than a single player, for both discounted games and finite horizon. This zero-sum game has a value by Sion's minimax theorem~\cite{sion1958}; the earlier theorems of Glicksberg~\cite{Glicksberg1952} and Fan~\cite{Fan1953} also apply. Discounting yields continuity of the utility function.

\begin{theorem}[Competing players explore less] \label{thm:competing_explore_less_discounted}
	Suppose arm $L$ has a known probability $p$ and arm $R$ has i.i.d.\ rewards with unknown success probability with prior $\mu$ (which is not a point mass). Assume that Alice and Bob are playing optimally in the zero sum game with discount factor $\beta$.
	
	Then there exists a threshold $p^* = p^*(\mu, \beta,-1) < g$, where $g = g(\mu,\beta)$ is the Gittins index of the right arm, such that for all $p > p^*$, with probability $1$ the players will not explore arm R. More precisely, define
	\begin{align} \label{def:p_star}
		p^* = p^*(\mu, \beta, \lambda) = \sup \Bigl\{p : \mbox{arm R is explored in some Nash equilibrium }\Bigr\}		
	\end{align}
	Then
	$p^*(\mu, \beta,-1) \leq (m \beta + g)/(1+ \beta)$, where $m$ is the mean of $\mu$.
\end{theorem}

\begin{remark}
	Note that $p^*(\mu, \beta,-1)$ tends to  $(m+g)/2$ as $\beta \to 1$.
\end{remark}

An analogue of Theorem \ref{thm:competing_explore_less_discounted} for finite horizon is shown in Section \ref{sec:competitive} (Theorem \ref{thm:competing_explore_less_finite}).

\medskip

In the single player one-armed bandit, obtaining information about the risky arm can compensate for a lower mean reward compared to the predictable arm. In the zero-sum case, the value of acquiring information is less clear since it can be copied by the opponent in the next round. The next theorem shows that such information still has value: competing players do not follow a myopic policy of pulling the arm with the highest mean reward in each round.

\begin{theorem}[Competing players are not completely myopic] \label{thm:non_myopic_discounted}
	In the setting of Theorem \ref{thm:competing_explore_less_discounted},
	let $\widetilde{p} = \widetilde{p}(\mu,\beta,\lambda)$ be the maximal threshold such that for all $p < \widetilde{p}$, with probability $1$ both players will explore arm $R$ in the initial round of any Nash equilibrium.
	Then $\widetilde{p}(\mu,\beta,-1) \ge m+\beta w/4$, where  $m$ is the mean of $\mu$ and $w$ is its variance.
\end{theorem}
\begin{remark}
	Note that $\widetilde{p} \leq p^*$ and we conjecture they are in fact equal.
\end{remark}


\subsection{Cooperating Players ($\lambda = 1$)}

In the fully cooperative model the players aim to optimize the same function, namely the sum of their rewards.
The players are allowed to agree on their strategies before play.

\begin{theorem}[Cooperating players explore more] \label{thm:cooperation}
	Consider two cooperating players, Alice and Bob, playing a one armed bandit problem with discount factor $\beta \in (0,1)$. The left arm has success probability $p$ and the right arm has prior distribution $\mu$ that is not a point mass.
	
	Then there exists $\widehat{p} > g= g(\mu,\beta)$, so that for all $p < \widehat{p}$, at least one of the players explores the right arm with positive probability under any optimal strategy pair maximizing their total reward.
\end{theorem}

We also show that cooperating players do not explore indiscriminately for every value of $p$. Proposition~\ref{prop:cooperating_careful} in Section \ref{sec:cooperative} makes this precise by showing a non-empty interval for $p$ in which cooperating players do not explore.

\subsection{Neutral Players ($\lambda = 0$)}

For neutral play the utility of each player is their total reward and the solution concepts will be Nash equilibrium and perfect Bayesian equilibrium.
Recall player $i$'s strategy $\sigma_i$ is a \emph{best response} to player $j$'s strategy $\sigma_j$ if no strategy $\sigma_i'$ achieves a higher expected utility against $\sigma_j$.
A mixed strategy profile $(\sigma_A, \sigma_B)$ is a \emph{Bayesian Nash equilibrium} if $\sigma_i$ is a best response for each player $i$. For brevity, we refer to such strategy profiles as Nash equilibria. 

\medskip

A \emph{Perfect Bayesian Equilibrium} is the version of subgame perfect equilibrium for games with incomplete information. A pair of strategies $(\sigma_A, \sigma_B)$ is a perfect Bayesian equilibrium if $(i)$ starting from any information set, subsequent play is optimal, and $(ii)$ beliefs are updated consistently with Bayes' rule on every path of play that occurs with positive probability.
Such equilibria are guaranteed to exist in this setting (see Fudenberg and Levine~\cite{fudenberg}).


\medskip

We will say that players learn from each other under some strategies if the total expected reward of each player is strictly higher than it would be for a single player using an optimal strategy. This can happen if the players infer additional information from each other's actions, beyond the bits that they observe themselves.

\begin{theorem}[Neutral players learn from each other] \label{neutral}
	Consider two neutral players, Alice and Bob, playing a one armed bandit problem with discount factor $\beta \in (0,1)$. The left arm has success probability $p$ and the right arm has prior distribution $\mu$ that is not a point mass. Then in any Nash equilibrium:
	\begin{enumerate}
		\item For all $p < g(\mu,\beta)$, with probability $1$
		at least one player explores. Moreover, the probability that no player explores by time $t$ decays exponentially in $t$.
		\item Suppose $p \in (p^*, g)$, where $p^*$ is the threshold above which competing players do not explore~\footnote{For the formal definition of $p^*$, see  Theorem~\ref{thm:competing_explore_less_discounted}.}. If the equilibrium is furthermore perfect Bayesian, then every (neutral) player has expected reward strictly higher than a single player using an optimal strategy.
	\end{enumerate}
\end{theorem}

We note that when $p > g$ there are perfect Bayesian equilibria in which the players do not explore; for example, always playing left is a perfect Bayesian equilibrium.

\medskip

A corollary of Theorem~\ref{neutral} is that in the range $(p^*, g)$ there is no perfect Bayesian equilibrium where the (neutral) players have the same trajectory with probability $1$.

\bigskip

Figure~\ref{fig:unified_plot} shows the different regions in which players explore depending on the success probability $p$ of the left arm as a function of the prior $\mu$ and the discount factor $\beta$.

\begin{figure}[h!]
	\centering
	\includegraphics[scale=0.9]{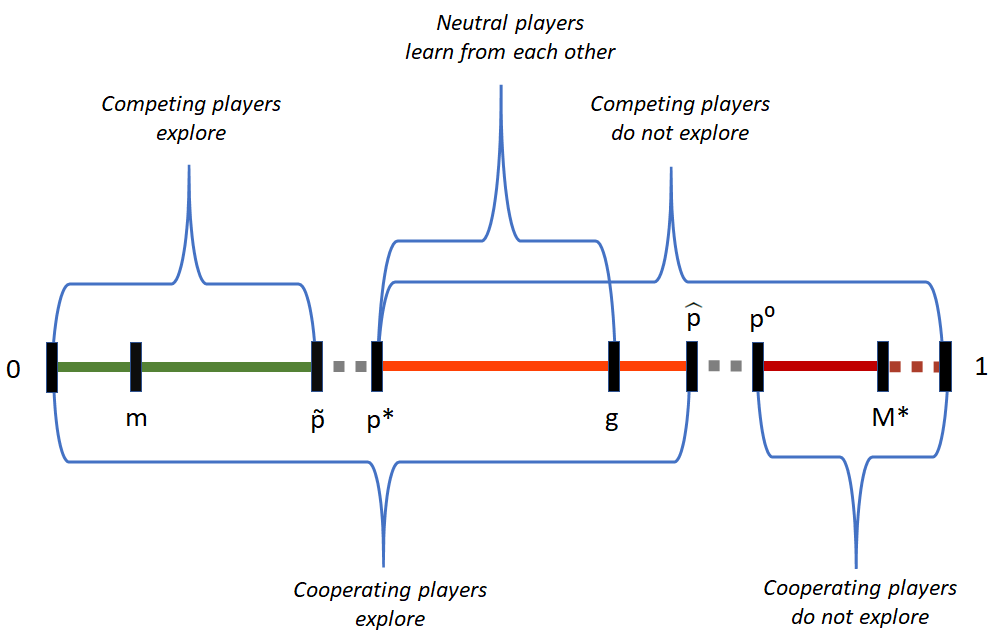}
	\caption{\emph{\small{Different regions in which players explore depending on the success probability $p$ of the left arm as a function of the prior $\mu$ and the discount factor $\beta$. Here $m$ is the mean of $\mu$, while $g = g(\mu, \beta)$ is the Gittins index of the right arm, $\widetilde{p}$ is the threshold where for all $p < \widetilde{p}$ competing players explore, $p^*$ the threshold where for all $p > p^*$ competing players do not explore, $\widehat{p}$ the threshold where for all $p < \widehat{p}$ cooperating players explore, and $p^{\circ}$ the threshold above which cooperating players do not explore. $M^*$ is the maximum of the support of $\mu$. Solid intervals are non-empty, dotted intervals may be empty. }}}
	\label{fig:unified_plot}
\end{figure}

\newpage

\subsection{Long Term Behavior}

We also analyze the long term behavior of the players and show that in every Nash equilibrium, competing and neutral players converge to playing the same arm forever from some point on. For $\lambda \not \in [-1,0]$, this can fail. The convergence proofs for competing and neutral players are technically very challenging and we explain the  intuition below.

\begin{theorem} \label{thm:neutral_competing_settle_same_arm}
	Consider two players, Alice and Bob, playing a one armed bandit problem with discount factor $\beta \in (0,1)$. The left arm has success probability $p$ and the right arm has prior distribution $\mu$ such that $\mu(p)=0$. Then in any Nash equilibrium, in both the competing ($\lambda = -1$) and neutral ($\lambda = 0$) cases, the players eventually settle on the same arm with probability $1$.
\end{theorem}

The intuition behind the proof for neutral players is that
if both players explore finitely many times, then we are done. Otherwise, there is a player, say Alice, who explores infinitely many times. Then Alice will eventually know which arm is better, so if she continues exploring we must have $\Theta > p$. Thus if Bob sees that Alice keeps exploring, he will eventually realize that $\Theta > p$ and will join her at the right arm.

Two obstacles make the proof delicate. The first is that the theorem applies to all Nash equilibria, without assuming subgame perfection, so we need to rule out non-credible threats  (which do occur in the case $\lambda = 1$, see Example~\ref{eg:uncredible_coop} below).
The second obstacle is that $\Theta$ might be very close to $p$, which delays the time at which Alice determines the better arm. This issue did not arise in Aoyagi's work since the distributions were discrete.
We overcome this obstacle by careful concentration arguments and Bayesian analysis; see sections \ref{app:concentration_lemma} and \ref{app:neutral_players}.

In the zero-sum setting there is an additional difficulty compared to the neutral case: the players might refrain from switching to the optimal arm in order to not trigger an adverse reaction from the opponent. The key to overcoming this difficulty is realizing that it is impossible for both players to benefit from repeatedly pulling the inferior arm. Each player might subjectively believe for some time that they are playing the better arm, but if a player keeps exploring, then their subjective evaluation of the risky arm will converge to the objective reality. The proof can be found in section ~\ref{app:settle_competing}.

\medskip

The next example shows that when $\lambda = 1$ there are Nash equilibria where aligned players do not settle on the same arm.

\begin{example}[Nash equilibria where players do not converge, $\lambda = 1$] \label{eg:uncredible_coop}
	Suppose Alice and Bob are aligned players in a one-armed bandit problem with discount factor $\beta$, where the left arm has success probability $p$ and the right arm has prior distribution $\mu$ that is a point mass at $m > p$.
	
	Then for every discount factor $\beta > 1/2$, there is a Nash equilibrium in which Bob visits both arms infinitely often.
\end{example}

Let $k \in \mathbb{N}$. Bob's strategy $S_B$ is to play left in rounds $0, k, 2k, 3k, \ldots$, and right in the remaining rounds.
Alice's strategy $S_A$ is to play right if Bob follows the trajectory above; if Bob ever deviates from $S_B$, then Alice switches to playing left forever. In Section~\ref{subsec:oscillations_nash} we show that if $k$ is large enough, then this is indeed a Nash equilibrium.
The strategies are depicted in Figure~\ref{fig:oscillating_cooperating}.

\begin{figure}[h!]
	\centering
	\includegraphics[scale = 0.7]{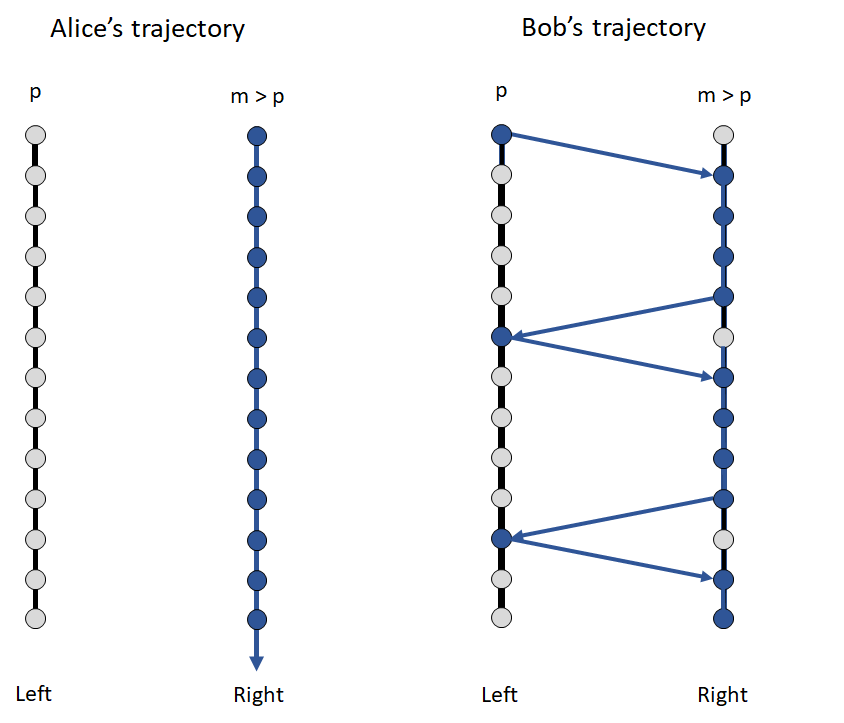}
	\caption{Trajectories of the players on the main line under strategies $(S_A, S_B)$ for $k=5$. The circles represent time units $0, 1, 2, \ldots$. The left arm has success probability $p$ and the right arm has prior distribution which is a point mass at $m > p$.}
	\label{fig:oscillating_cooperating}
\end{figure}


Note this Nash equilibrium is not Pareto optimal: both players could improve by a joint deviation to always playing right. This Nash equilibrium is also not subgame perfect since it has a non-credible threat by Alice (of playing left forever after any deviation by Bob).

\subsection{Roadmap of the Paper}

A more detailed discussion of related work is given in section~\ref{sec:related}. Background on the Gittins index is given in section~\ref{sec:background_gittins_nash}.

Section~\ref{sec:competitive} has the analysis for competing  players, showing that they explore less than a single player, but are also not completely myopic.
Section~\ref{sec:cooperative} studies fully cooperating players, showing they explore more than a single player.
Section~\ref{sec:neutral} studies neutral players, showing that they learn from each other.

The long term behavior of the players is studied in section~\ref{app:long_term}, where we show that both neutral and competing eventually settle on the same arm in every Nash equilibrium, while this can fail for cooperating players.

Directions for future research are described in section~\ref{sec:future_directions}.

\section{Related Work} \label{sec:related}


This work is related to several streams of research.
Bandit learning problems with multiple players have been studied in the collision model, where players are pulling arms independently. The players are cooperating---trying to maximize the sum of rewards---and can agree on a protocol before play, but cannot communicate during the game. Whenever there is a collision at some arm, then no player that selected that arm receives any reward. This is motivated by applications such as cognitive radio networks, where interference at a channel destroys the signal for all the players involved.

Several research directions in the setting with collisions include designing algorithms that allow the players to maximize the total reward, depending on whether the environment is adversarial (see, e.g., \cite{ALK19,BLPS19}) or stochastic (e.g.,  \cite{LJP08,LZ10,Kalathil14,LM18,BL18}), and on whether they receive feedback about the collision or not~\cite{AM14,RSS16,BBMKP17,BP18}. Hillel, Karnin,  Koren, Lempel, and Somekh~\cite{NIPS2013_4919} study exploration in multi-armed bandits in a setting with pure exploration, where players collaborate to identify an $\epsilon$-optimal arm.

Aoyagi~\cite{aoyagi,aoyagi_corr} studies bandit learning with the same feedback model as ours, when there are multiple neutral players and two risky arms.
The main result is that assuming a property on the distributions of the arms, in any Nash equilibrium the players settle eventually on the same arm. This is related to Aumann's agreement theorem~\cite{Aumann76}, which shows in a Bayesian setting that rational players cannot agree to disagree. While this game is not captured by the formal model of the theorem in~\cite{Aumann76}, the result is similar conceptually.

\smallskip

Another natural feedback model is perfect monitoring, where all players see the actions and rewards of each other. This is not interesting in the zero sum case (see Remark~\ref{remark:perfect_monitor_myopic}), but in the neutral and cooperative cases it is close to classic papers in the economics of strategic experimentation.
Bolton and Harris~\cite{BH99} study a multiplayer learning problem in
in the one-armed bandit model with continuous time and perfect monitoring, where the feedback model is that players can observe each other's past actions and rewards. The main effects observed in symmetric equilibria are a free rider effect and an encouragement effect, where a player may explore more in order to encourage further exploration from others. Cripps, Keller, and Rady~\cite{keller2005strategic} characterize the unique Markovian equilibrium of the game. They also show that asymmetric equilibria are more efficient than symmetric ones, as it is more useful for the players to take turns experimenting.

Heidhues, Rady, and  Strack~\cite{HRS15} study the discrete version of this model and establish that in any Nash equilibrium, players stop experimenting once the common belief falls below a single-agent cut-off. They also show that the total number of experiments performed in equilibrium differs from the single-agent optimum by at most one. The work in~\cite{HRS15} additionally studies a model with imperfect monitoring, where the players can observe each other's actions but only their own rewards. The players can also communicate via cheap talk. One of the main findings is that cheap talk is incentive compatible and the socially optimal symmetric experimentation profile can be supported as a perfect Bayesian equilibrium. Klein and  Rady~\cite{KR11} study the one-armed bandit model with public monitoring, where the correlation across bandits is negative.

Rosenberg, Solan, and Vieille~\cite{social_learning} and Rosenberg,  Salomon, and Vieille~\cite{RSV13} study the model with imperfect monitoring, 
except the decision to switch from the risky arm to the safe one is irreversible. As they write, \emph{``Dropping the assumption that payoffs are publicly observed raises new issues. Player $i$ would like to make inferences about player $j$'s observations on the basis of player $j$'s actions, but cannot do so without knowing how player $j$'s decisions relate to player $j$'s observations, that is, $j$'s strategy''.}
 Mansour, Slivkins, and Wu~\cite{MSW18} consider a setting where players arrive and depart over time.

\smallskip

Another line of work in economics studies the interplay between competition and innovation (see, e.g.,  \cite{AJ88}).
Close to our high level direction is the work of Besanko and  Wu~\cite{BW13}, which studies the tradeoff between cooperation and competition in R\&D via learning in a one-armed bandit model, where the safe arm is the established product and the risky arm is the novel product. Our model is also related to the literature on stochastic games with imperfect monitoring (see, e.g., \cite{APS90}). Rosenberg, Solan, and Vieille~\cite{RSV09} study a model of dynamic games with informational externalities and analyze the rate of experimentation, providing conditions under which players eventually reach a consensus. The work in ~\cite{RSV09} is also related to the Rothschild conjecture.

\smallskip

In evolutionary biology, there is a line of research dedicated to understanding the extremely high levels of cooperation observed in social insects (see, e.g., \cite{Anderson84} and \cite{Boomsma07}), including designing mathematical models to explain why eusociality would evolve from natural selection (e.g.,  \cite{NTW10}). Social learning was studied in the context of understanding altruistic behaviors such as sharing of information about food locations (see, e.g., \cite{GLR10}).

\smallskip

The theme of incentivizing exploration was studied, for example, by Kremer, Mansour and Perry~\cite{KMP13}, and
Frazier, Kempe, Kleinberg, and Kleinberg~\cite{FKKK14}, where the problem is that a principal wants to explore a set of arms, but the exploration is done by a stream of myopic agents that have their own incentives and may prefer to exploit instead. Mansour, Slivkins, and Syrgkanis~\cite{MSS15} design Bayesian incentive compatible mechanisms for such settings.

Aridor, Liu, Slivkins, and Wu~\cite{ALSW19} empirically study the interplay between exploration and competition in a model where multiple firms are competing for the same market of users and each firm commits to a multi-armed bandit algorithm. The objective of each firm is to maximize its market share and the question is when firms are incentivized to adopt better algorithms. Multi-armed bandit problems with strategic arms have been studied theoretically by  Braverman, Mao, Schneider, and Weinberg \cite{BMSW19}, in the setting where each arm receives a reward for being pulled and the goal of the principal is to incentivize the arms to pass on as much of their private rewards as possible to the principle.

\smallskip

Immorlica, Kalai, Lucier, Moitra, Postlewaite, Tennenholtz~\cite{IKLMPT11} study competitive versions of classic search and optimization problems by converting them to zero-sum games. One of the open questions from~\cite{IKLMPT11} was whether competition between algorithms improves or degrades expected performance in that framework, which was answered by Dehghani, Hajiaghayi, Mahini, and Seddighin~\cite{DHMS16} for the ranking duel and a more general class of dueling games.

Finally, we note that a conference version of the present paper appeared in COLT 2021~\cite{YP21}.

\section{Background on the Gittins index}  \label{sec:background_gittins_nash}


In this section we give a brief background on the Gittins index.
The Gittins index was first defined by  \cite{bradt_johnson_karlin} and its importance for the multi-armed bandit problem was shown by  \cite{gittins_jones}; see also \cite{Gittins79,berry_book}.
The description below    mainly follows \cite{Gittins_index_background_Weber}.

In the classic multi-armed bandit problem, there is a gambler who can play any of $n$ one-armed bandit machines. The goal of the gambler is to play in a way that maximizes its expected total discounted reward.

Let $[n] = \{1, \ldots, n\}$ be the set of bandits, each of which is  a Markov process. The state of bandit $j$ at time step $t \in \{0,1,\ldots, \}$ is denoted $x_j(t)$. When playing bandit $j$, the gambler receives reward $R_j(x_j(t))$ and the state of bandit $j$ changes in a known Markov fashion, while the states of the other bandits remain unchanged.

A policy stipulates which bandit to play next, given the history of play and the rewards obtained so far. Given policy $\pi$, let $j(t)$ denote the bandit played at time step $t$.
The goal is to find a policy that maximizes the expected discounted reward, defined as  follows:
\begin{align}
	V_{\pi}(x) = \Ex_{\pi}\left[ \sum_{t=0}^{\infty} \beta^t \cdot R_{j(t)}\bigl(x_j(t)\bigr) \mid x(0) = x \right]
\end{align}

The solution to this optimization problem is described by functions $G_j$, which are known as Gittins indices. Each function $G_j$ only depends on the current state bandit $j$. Gittins and Jones~\cite{gittins_jones} proved  that playing bandit $j$ at time $t$ is optimal if and only if
\begin{align}
	G_j(x_j(t)) = \max_{1 \leq i \leq n} G_i\bigl(x_i(t)\bigr)\,.
\end{align}

The Gittins index has several equivalent definitions. The most convenient for us is the {\em retirement value} of the arm. Suppose at every step the Gambler may choose  to ``retire'' and receive a payment $p$ in the current step and in all subsequent steps.  Alternatively, he may select arm $j$ and receive the current reward at that arm, while
maintaining the option to retire at any point in the future.
Given that arm $j$ is currently at state $x_j$, the Gittins index $G(x_j)$ is the infimum of the values  $p$ for which retirement now is preferable.


\paragraph{Bernoulli bandits} In the paper, we focus on a special case known as Bernoulli bandits, where the rewards are $1$ (successes) or $0$ (failures). Arm $j$ has a known prior $\mu_j^0$ on $[0,1]$; its success probability $\Theta_j$ (unknown to the player) is drawn from $\mu_j$, and each time arm $j$ is selected, the reward is 1 with probability $\Theta_j$, independently of previous picks.   The state of arm $j$ at time $t$ is described by a pair $(s_j(t), f_j(t))$, where $s_j(t)$ and $f_j(t)$ are the number of successes and failures, respectively, obtained at arm $j$ until time $t$. The posterior distribution of the success probability $\Theta_j$ after step $t$ is
$\mu_j^t$ which has density proportional to $   \theta^{s_j(t)} (1-\theta)^{f_j(t)}$ with respect to $\mu_j^0$, i.e., for Borel sets $A \subset [0,1]$,
$$\mu_j^t(A)= \frac{\int_A \theta^{s_j(t)} (1-\theta)^{f_j(t)} \, d \mu_j^0(\theta)}
{\int_0^1 \theta^{s_j(t)} (1-\theta)^{f_j(t)} \, d \mu_j^0(\theta)} \,.
$$

If the player selects arm $j$ at time $t+1$, then the expected reward from that move (given the history) is the mean of the posterior
$$\int_0^1 \theta   d\mu_j^t(\theta) \,.$$
That expected reward is also the transition probability  from the state $ (s_j(t),f_j(t))$ to  the state
$(s_j(t)+1,f_j(t))$. The transition probability  to $(s_j(t),f_j(t)+1)$
is the complementary probability
$\int_0^1 (1-\theta) d\mu_j^t(\theta) \,.$

In the case we focus on in this paper, there are just two arms, the left arm has success probability $p$ and the right arm has a nonatomic prior $\mu$ on the (unknown) success probability. In this case,
Let
\begin{equation} \label{defm1}  m_1 := \frac{1}{m} \int_{0}^{1} x^2 \cdot \mathop{d \mu}
\end{equation}
 denote the posterior mean at the right arm after observing $1$ in round zero, and
 \begin{equation} \label{defw}
 w := \int(x-m)^2 \mathop{d \mu}= m \cdot (m_1 - m)
\end{equation}
denotes the variance of the right arm.

\section{Competitive Play} \label{sec:competitive}

In this section we study the zero-sum game, corresponding to $\lambda = -1$.
This is an extensive-form game with an initial move by Nature; see ~\cite{maschler_book,KP17}. 
The game has a value by Sion's minimax theorem~\cite{sion1958}; moreover, the value is zero by symmetry.

Recall that in the discounted setting, the {Gittins index} $g = g(\mu, \beta)$ of the right arm is defined as the infimum of the success probabilities $p$ where playing always left is optimal for a single player.

\begin{observation} \label{obs1}
	If selecting L in every round is the only optimal strategy for a single player, then in every optimal strategy for Alice in the zero sum game, she never explores, and similarly for Bob.
\end{observation}
\begin{proof}
	If Alice explores with positive probability and Bob never explores, then her expected net payoff will be negative by the hypothesis, since she is not learning anything from Bob's actions.
\end{proof}

We also give a simple lower bound on the Gittins index, the proof of which is included for completeness. Recall that for an arm with prior $\mu$, we have $g = g(\mu,\beta)$ is the Gittins index of the arm, $m$ is the mean, and $w$ is the variance of $\mu$.
\begin{lemma} \label{lem:gbound}
	Consider one arm with prior $\mu$. Then
	$g(\mu,\beta) \geq m + \beta  w/2$.
\end{lemma}
\begin{proof}
	Suppose Alice is playing the one-armed bandit game by herself where  the right arm has distribution $\mu$ that is not a point mass and left arm has known probability $p = g = g(\mu,\beta)$. Consider the following strategy for Alice:
	\begin{itemize}
		\item Round zero: play right.
		\item Round one: play left if $0$ was observed in round zero, and right if $1$ was observed.
		\item Round two onwards: play left.
	\end{itemize}
	Then by the definition of $g$, we have that this strategy for Alice is at most as good as retiring and receiving $g$ forever, and so
	\begin{align} \label{eq:g_bound}
		\frac{g}{1-\beta} \geq m + (1 - m) \cdot \frac{g  \beta}{1-\beta} + m\left(m_1  \beta + \frac{g  \beta^2}{1-\beta}\right)
	\end{align}
	Recall that $m \cdot m_1 = m^2 + w$. Using this in (\ref{eq:g_bound}) and rearranging, we obtain:
	\begin{align} \label{eq:g_lower_bound}
		& g \geq m + \frac{\beta w}{1 + m \beta}
		\geq m + \frac{\beta w}{2}
	\end{align}
\end{proof}



Our first theorem shows that when using optimal strategies, competing players explore less than a single player, for both discounted games and finite horizon.

\bigskip

\noindent \textbf{Theorem} \ref{thm:competing_explore_less_discounted} [Competing players explore less] (restated).  \emph{Suppose arm $L$ has a known probability $p$ and arm $R$ has i.i.d.\ rewards with unknown success probability with prior $\mu$ (which is not a point mass). Assume that Alice and Bob are playing optimally in the zero sum game with discount factor $\beta$. \\
	Then there exists a threshold $p^* = p^*(\mu, \beta,-1) < g$, where $g = g(\mu,\beta)$ is the Gittins index of the right arm, such that for all $p > p^*$, with probability $1$ the players will not explore arm R. More precisely, define
	\begin{align} \label{def:p_star}
		p^* = p^*(\mu, \beta, \lambda) = \sup \Bigl\{p : \mbox{arm R is explored in some Nash equilibrium }\Bigr\}		
	\end{align}
	Then
	$p^*(\mu, \beta,-1) \leq (m \beta + g)/(1+ \beta)<g$, where $m$ is the mean of $\mu$.}

\begin{remark}
	Note that the upper bound on $p^*(\mu, \beta,-1)$ tends to $(m+g)/2$ as $\beta \to 1$.
\end{remark}

We establish the theorem by showing that for $p$ close enough to $g$, a player using an optimal strategy is never the first to explore.
However, each player will need as part of their strategy a contingency plan for what to do if the other player deviates from the main path. In fact, for $p < g$, the strategy of always playing left is not part of any equilibrium; if Bob always plays left, then Alice can play the one player optimum (of pulling the arm with the highest Gittins index) and win.

\medskip

\begin{proof}[Proof of Theorem~\ref{thm:competing_explore_less_discounted}]
Given a strategy pair $(S_A,S_B)$ for Alice and Bob, the utility to Alice is the random variable
$$u_A(S_A,S_B):=  \Gamma_A(S_A,S_B)-\Gamma_B(S_A,S_B)  \,.$$
Due to symmetry, the value of the game
$\min_{S_B} \max_{S_A}  \Ex u_A(S_A,S_B)$ equals zero.

Moreover, once  we fix a strategy $S_B$ for Bob,  the expected utility $\Ex u_A(S_A,S_B)$ to Alice is a convex combination of expected utilities $\Ex u_A(S_A^{\rm pure},S_B)$ obtained when Alice plays a pure strategy $S_A^{\rm pure}$ in the support of $S_A$. Thus Alice has a pure best response $S_A^{\rm pure}$ to Bob's strategy $S_B$, so:
\begin{equation}
\label{avalue1}
\min_{S_B} \max_{S_A}  \Ex u_A(S_A,S_B)=   \min_{S_B} \max_{S_A^{\rm pure}} \Ex u_A (S_A^{\rm pure},S_B) \,,
\end{equation}

If $p>g$, then Observation \ref{obs1} implies that both players never explore when using optimal strategies. So we may and shall assume that $p \le g$. We will show that if 
\begin{align}  \label{assume-exp}
\bullet \; \;	& (S_A^*,S_B^*) \; \text{is a pair of optimal strategies in which at least one of the players explores} \notag \\
	& \text{with positive probability,}
\end{align}
then $ p \le (m \beta + g)/(1+ \beta)$. 

Denote by $D_k^A$ the event that neither player explores before round $k$, and Alice explores at round $k$. Let $D_k^B$ denote the corresponding event with Bob in place of Alice.  (Note these events are not necessarily disjoint.) Finally, let $D_\emptyset$ be the event that neither player ever explores.

The assumption \eqref{assume-exp} implies that at least one of $D_k^A$ and $D_k^B$ has positive probability for some $k$. Without loss of generality, we may assume it is $D_k^A$.

Consider the following pure strategy $S_B$ for Bob:
\begin{itemize}
	\item Play left until Alice selects the right arm, say in some round $k$. Then play left again in round $k+1$, and then starting with round $k+2$, copy Alice's move from the previous round. In particular, under $S_B$, Bob never explores (i.e., plays right) first.
\end{itemize} Then by optimality of $S_A^*$,
\begin{equation}
\label{avalue2}
\Ex u_A(S_A^*,S_B^{ }) \ge
\Ex u_A(S_A^*,S_B^*)= 0.
\end{equation}
When Alice plays $S_A^*$ against $S_B$, the event $D_k^A$ has  positive probability  whenever this holds under $(S_A^*,S_B^*)$, because Alice  cannot distinguish $S_B^{}$ from $S_B^*$ until that first exploration.
Therefore
\begin{align}
	0 & \le \Ex u_A(S_A^*,S_B^{}) \notag \\
	& = \mathbb{P}_{S_{A}^{*}, S_{B}^{ }}(D_\emptyset) \cdot \Ex \left[u_A(S_A^*,S_B^{})|D_\emptyset \right] + \sum_{k=0}^\infty \mathbb{P}_{S_{A}^{*}, S_{B}^{ }}(D_k^A)
	 \Ex \left[u_A(S_A^*,S_B^{})|D_k^A \right] \,. \notag
\end{align}

We have $\Ex [u_A(S_A^*,S_B)|D_\emptyset]=0$ since neither player explores under $D_{\emptyset}$, and \\$\sum_{k=0}^\infty \mathbb{P}_{S_{A}^{*}, S_{B}^{ }}(D_k^A) >0$. Then there must exist some integer $k \ge 0$ such that
\begin{align}
	\mathbb{P}_{S_{A}^{*}, S_{B}^{ }}(D_k^A)>0 \; \text{ and } \; \Ex \left[u_A(S_A^*,S_B^{})|D_k^A\right] \ge 0\,.
\end{align}
Therefore, in the support of $S_A^*$ there exists a pure strategy $S_A$ for Alice such that $S_A$ explores first at  round $k$ against $S_B$, and
\begin{equation}
\label{avalue3}
\Ex u_A(S_A,S_B) =\Ex \bigl(\Gamma_A(S_A,S_B)-\Gamma_B(S_A,S_B)\bigr ) \ge 0.
\end{equation}
	
	Then Alice's total reward $\Gamma_A = \Gamma_A(S_A, S_B)$ has expectation
	$$
	\Ex(\Gamma_A) =  \sum_{t=0}^{k-1} p \cdot \beta^t + \sum_{t = k}^{\infty} \Ex(\gamma_A(t)) \cdot \beta^t\,.
	$$
Here and in the rest of the proof, we assume that Alice and Bob employ the pair of pure strategies $(S_A,S_B)$ described above.
	Bob's total reward $\Gamma_B = \Gamma_B(S_A, S_B)$ has expectation
	\begin{align}
		\Ex(\Gamma_B) & = \sum_{t = 0}^{k+1} p \cdot \beta^t +
		\sum_{t =k+ 2}^{\infty} \Ex(\gamma_B(t)) \cdot \beta^t
		= \sum_{t = 0}^{k+1} p \cdot \beta^t + \sum_{t = k+2}^{\infty} \Ex(\gamma_A(t-1)) \cdot \beta^t \notag \\
		& =  \sum_{t = 0}^{k+1} p \cdot \beta^t + \sum_{t=k+1}^{\infty} \Ex(\gamma_A(t)) \cdot \beta^{t+1}  \,.
	\end{align}
	
	Note that $\Ex(\gamma_A(k)) = m$. Then
	the difference in expected rewards from the strategy pair $(S_A,S_B)$  is
	\begin{small}
	\begin{align}
		\Ex(\Gamma_A) - \Ex(\Gamma_B) & = \left(\sum_{t=0}^{k-1} p \cdot \beta^t + m \cdot \beta^k + \sum_{t = k+1}^{\infty} \Ex(\gamma_A(t)) \cdot \beta^t\right)
		- \left( \sum_{t = 0}^{k+1} p \cdot \beta^t + \sum_{t=k+1}^{\infty} \Ex(\gamma_A(t)) \cdot \beta^{t+1}\right) \notag \\
		& =  (m\beta - p - p \beta) \beta^k+ (1-\beta) \cdot \left( m \beta^k + \sum_{t=k+1}^{\infty} \Ex(\gamma_A(t)) \cdot \beta^t \right)  \,.
	\end{align}
\end{small}
	Since Bob is copying Alice, she is not learning anything from his actions, so her total reward from round $k$ to $\infty$ is at most the maximum reward that a single player can obtain, namely $g/(1-\beta)$.
	Thus we can bound the difference between the players by
	\begin{align} \label{eq:delta_gittins_myopic}
		\Ex(\Gamma_A) - \Ex(\Gamma_B) \leq \Bigl(m \beta - p (1 + \beta)  + g\Bigr) \beta^k\,.
	\end{align}
	
	In view of \eqref{avalue3}, the inequality (\ref{eq:delta_gittins_myopic}) implies that
	\begin{align}
		m \beta - p (1 + \beta) + g \ge 0 \,, \; \, \text{i.e.,} \; \, p  \le  \frac{m \beta + g}{1 + \beta}\,.
	\end{align}
  Thus $p^* \leq (m \beta + g)/(1 + \beta)$ as required.
\end{proof}

We also show that competition reduces exploration for a large finite horizon.


\begin{theorem}[Competing players explore less, finite horizon] \label{thm:competing_explore_less_finite}
	Suppose arm $L$ has a known probability $p$ and arm $R$ has a known distribution $\mu$ with mean $m$. Let $T$ be the horizon and $M^*$ the maximum of the support of $\mu$. We have
	\begin{enumerate}
		\item If $p > (m + M^*)/2$, then the players do not explore arm $R$.
		\item However, for every $p < M^*$, if $T$ is large enough, then a single player will explore given the same two arms.
	\end{enumerate}
\end{theorem}
\begin{proof}
	For part 1 of the statement, we use the same strategy $S_B$ for Bob as in the proof of Theorem~\ref{thm:competing_explore_less_discounted}:  play left until Alice selects the right arm, say in some round $k$. Then play left again in round $k+1$, and then starting with round $k+2$ copy Alice's move from the previous round.
	
	Fix an arbitrary pure strategy $S_A$ for Alice. If $S_A$ never plays right first we are done. If Alice plays right for the first time in round $k$, then $\gamma_A(k) = m$. Her expected total reward can be written as
	$$
	\Ex(\Gamma_A) = \left(\sum_{t = 0}^{k-1} p \right) + m + \sum_{t = k+1}^{T} \Ex(\gamma_A(t))\,.
	$$
	Since from round $k+2$ Bob is copying Alice's previous move, we have $$\Ex(\gamma_B(t)) = \Ex(\gamma_A(t-1)) \mbox{ for all } t \geq k+2\,.$$
	Then Bob's total expected reward is
	$$
	\Ex(\Gamma_B)  = p \cdot (k+2) + \sum_{t = k+2}^{T} \Ex(\gamma_B(t)) = p \cdot (k+2) + \sum_{t = k+1}^{T-1} \Ex(\gamma_A(t))\,. 
	$$
	
	
	If $k = T$, the difference in rewards is $$\Ex(\Gamma_A) - \Ex(\Gamma_B) = \Ex(\gamma_A(T)) - \Ex(\gamma_B(T)) =  m - p < 0\,.$$
	
	If $k \leq T-1$, by choice of $p > (m+M^*)/2$, the difference in rewards is bounded by $$\Ex(\Gamma_A) - \Ex(\Gamma_B) = (m - p) + \Bigl(\Ex(\gamma_A(T)) - p\Bigr)
	\leq m - 2p + \max\{M^*, p\} < 0\,.$$
	Thus Alice loses in expectation.
	
	
	\medskip
	
	Part 2 of the statement holds since a single player (say Bob) optimizing his expected total reward has an algorithm with sublinear regret. Formally, the expected mean of the best arm is
	$$
	\xi = p \cdot \mu(0,p) + \int_{p}^{1} x \mathop{d\mu(x)} \,.
	$$
	If $p < M^*$, then $\xi > p$. Bob's expected reward when playing optimally can be bounded from below by using a low regret algorithm (see, e.g.,~\cite{BCB12}), which gives
	$$
	\Ex(\Gamma_B) \geq \xi \cdot (T+1) - C \cdot \sqrt{T+1}, \; \mbox{for some constant } C \geq 0
	$$
	We want that
	$$
	\xi \cdot (T+1) - C \cdot \sqrt{T+1} > p \cdot (T+1)
	$$
	This will hold if $$T> \frac{C^2}{(\xi - p)^2}\,.$$
	For such values of $T$, a regret minimizing algorithm will do better than playing left forever (in expectation).
\end{proof}

\begin{remark} \label{rmk:finite_index}
	Part 2 of Theorem~\ref{thm:competing_explore_less_finite} is due to 
	\cite{bradt_johnson_karlin} (see page 1073) and included for comparison. We can define in the finite horizon setting an index $g_T = g(\mu, T)$ as the infimum of $p$ where playing always left is optimal for a single player. Part 2 of Theorem~\ref{thm:competing_explore_less_finite} is equivalent to
	$\lim_{T \to \infty} g_T = M^*$.
	A similar proof shows that for any prior $\mu$, we have
	$\lim_{\beta \to 1} g(\mu, \beta) = M^*$.
\end{remark}


In the single player one-armed bandit, obtaining information about the risky arm can compensate for a lower mean reward compared to the predictable arm. In the zero-sum case, the value of acquiring information is less clear since it can be copied by the opponent in the next round. The next theorem shows that such information still has value: competing players do not follow a myopic policy of pulling the arm with the highest mean reward in each round.

\medskip

\noindent \textbf{Theorem} \ref{thm:non_myopic_discounted} [Competing players are not completely myopic] (restated). \emph{In the setting of Theorem \ref{thm:competing_explore_less_discounted}, let $\widetilde{p} = \widetilde{p}(\mu,\beta,\lambda)$, be the maximal threshold such that for all $p < \widetilde{p}$, with probability $1$ both players will explore arm $R$ in the initial round of any Nash equilibrium.\\
Then $\widetilde{p}(\mu,\beta,-1) \ge m+\beta w/4$. (Recall that  $m = \int_{0}^{1} x \mathop{d\mu}(x)$ and $w = \int_{0}^{1} (x-m)^2 \mathop{d \mu(x)}$.)}\\

	Note that $\widetilde{p} \leq p^*$ and we conjecture they are in fact equal.

\begin{remark} \label{remark:perfect_monitor_myopic}
	In the feedback model with perfect monitoring, where players observe each other's actions and rewards, the myopic strategy of selecting the arm with the highest current mean (in each round) is optimal. The reason is that exploration gives no future (informational) advantage to the exploring player.
\end{remark}

\begin{proof}[Proof of Theorem~\ref{thm:non_myopic_discounted}]
	Let $S_A$ and $S_B$ denote strategies for Alice and Bob, respectively, with the constraint that in round zero, Alice plays $R$ and Bob plays the predictable arm $L$. Define
	$$v_0^{(p)} = v_0 = \sup_{S_A} \, \inf_{S_B} \,  \Ex(\Gamma_{A}(S_A, S_B) - \Gamma_{B}(S_A, S_B)).$$
	We will show   that for $p <m + \beta  w/4$, we have $v_0^{(p)} > 0$.

	Take any mixed strategy of Alice in which she plays R in round zero and
	consider what happens from round one onwards. First recall that if Alice saw 1 in round zero, then the posterior mean is $$m_1 = \frac{1}{m} \cdot \int_{0}^{1} x^2 \cdot  \mathop{d\mu(x)}.$$ If Alice saw 0, then the posterior mean is $$m_0 = \frac{1}{1-m} \cdot \int_{0}^{1} x(1-x) \cdot  \mathop{d\mu(x)}.$$
	Note that $m \cdot (m_1 - m)= (1-m)\cdot (m-m_0) = w.$
	
	Fot $t=0,1,2$, Let $v_t(x,y) = v_t^p(x,y)$ denote the net value for Alice (= maxmin over mixed strategies) from the beginning of round $t$, when in round 1 she plays $x$ and Bob plays $y$.
	Here, since we will let Alice declare her strategy first, we have $y \in \{L, R\}$ (that is, $y = L$ means that Bob plays left in round 1 no matter what; similarly $y = R$ means that he plays right in round 1) and $x \in \{L, R, S\}$, where $S$ is the following Alice round 1 strategy:
	\begin{itemize}
		\item Play left upon seeing 0 in round zero.
		\item Play right upon seeing 1 in round zero.
	\end{itemize}
	
	Write $\delta = p - m \geq 0$. Then $v_0(x,y) = - \delta + v_1(x,y)$. Let $\tilde{v}_t(x,y)$ be defined in the same way as $v_t(x,y)$ when Bob knows Alice's record in round zero. Clearly, $\tilde{v}_t(x,y) \leq v_t(x,y)$, for all strategies $x,y$.
	
	By comparing the information available to the players based on the strategies played in rounds 0 and 1, observe:
	\begin{itemize}
		\item $v_2(x, L) \geq 0$ for all $x \in \{L, R, S\}$.
		\item $v_2(R, y) \geq 0$ for all $y\in \{L, R\}$.
	\end{itemize}
	Next we bound Alice's net value at round zero if Bob plays L and Alice plays strategy S:
	\begin{align} \label{eq:v0_SL_bound}
		v_0(S, L) = -\delta + \beta m(m_1 - p)+ v_2(S, L)
		\geq - \delta + \beta m(m_1 - m - \delta) 
		= \delta(-\beta  m - 1) + \beta  w\,.
	\end{align}
	If Bob plays L in round one and Alice plays R in round one no matter what bit she observed at round zero, then we have:
	\begin{align} \label{eq:v0_RL_bound}
		v_0(R, L) & = - \delta + \beta \cdot (- \delta) + v_2(R, L) \geq - 2 \delta\,.
	\end{align}
	If both players play right in round one we get
	\begin{align} \label{eq:v0_RR_bound}
		v_0(R, R) = - \delta + v_2(R, R)\,.
	\end{align}
	We establish a few facts about the net gains of the players.
	
	\begin{lemma}
		$\tilde{v}_2(L, R) = - v_2(R, R)$.
	\end{lemma}
	\begin{proof}
		Recall that by definition of $\tilde{v}_2(L, R)$, Bob knows two results from the right arm, while Alice knows only one. For $v_2(R,R)$ this is reversed, since in this case Alice knows two results from the right arm while Bob knows only one.
	\end{proof}
	\begin{lemma}
		$v_2(S,R) \geq \tilde{v}_2(S,R) \geq \tilde{v}_2(L,R) = -v_2(R,R)$.
	\end{lemma}
	\begin{proof}
		Note the first inequality holds since for $\tilde{v}_2(S,R)$ Alice tells Bob the bit she saw in round zero, while in the case of $v_2(S,R)$ she does not. The second inequality holds since by playing strategy $S$ Alice has at least as much information as when playing $L$ in round one, so she will do at least as well under $S$ as she would do under $L$ from round two onwards.
	\end{proof}
	By decomposing $v_0(S,R)$ into the payoff obtained in rounds zero, one, and the payoff from round two onwards, we obtain
	\begin{align} \label{eq:v2_RR_bound}
		v_2(R,R) + v_0(S,R) & \geq - \delta + (1-m)  \beta  (p - m_0) \notag \\
		& = - \delta + \beta  (1-m)  (m - m_0 + \delta) \notag \\
		& = \delta  \left[\beta  (1-m) - 1 \right] + \beta  w
	\end{align}
	Denote by $\frac{R+S}{2}$ the mixed strategy of Alice for round one in which she plays strategy $R$ with probability $1/2$ and strategy $S$ with probability $1/2$.
	Averaging (\ref{eq:v0_RR_bound}) and (\ref{eq:v2_RR_bound}), we have
	\begin{align} \label{eq:v0_mixed_bound_R}
		v_0\left(\frac{R+S}{2}, R\right) \geq \delta  \left[ \frac{\beta  (1-m)}{2} - 1 \right] + \frac{\beta w}{2} \geq \frac{\beta w}{2} - \delta
	\end{align}
	Averaging (\ref{eq:v0_RL_bound}) and (\ref{eq:v0_SL_bound}), we obtain
	\begin{align} \label{eq:v0_mixed_bound_L}
		v_0\left(\frac{R+S}{2}, L\right) \geq \delta  \left[ \frac{-\beta  m}{2} - 3/2 \right] + \frac{\beta  w}{2} \geq \frac{\beta  w}{2} - 2\delta  \,.
	\end{align}
	Thus   for $p<m+\beta w/4$, we have $v_0^{(p)} > 0$.	To complete the proof, consider any (unconstrained) strategy pair $(S_A,S_B)$ that form an equilibrium.
	Since the game has value zero, these strategies must yield the same expected payoff to the two players. If (aiming for a contradiction)  at least one of the strategies (e.g, $S_B$) assigns positive probability to playing $L$ in move 0, then by playing $R$ in   move 0, and continuing optimally, Alice could ensure that her (net) utility is  strictly positive
	(Since on the event that Bob's initial move is $R$ their expected payoff's are equal, and on the complementary event she receives a strictly greater expected payoff.)
	This improvement shows that such a strategy pair cannot be an equilibrium.
\end{proof}

In Section \ref{app:competitive} we show improved bounds for a uniform prior and a plot with the bounds.

\section{Cooperative Play} \label{sec:cooperative}

In this section we study the scenario of $\lambda = 1$, where the players are cooperating.
Cooperating players aim to maximize the sum of their rewards and can agree on their strategies before the game starts. Later the players cannot communicate beyond seeing each other's actions.

We show that when the players are cooperating they explore even more than a single player. In particular, they explore even if the known arm has a probability $p$ higher than the Gittins index of the right arm.

Recall the definitions of $m_1 = \frac{1}{m} \int_{0}^{1} x^2 \cdot \mathop{d \mu}$ as posterior mean at the right arm after observing $1$ in round zero, and of $w = m \cdot (m_1 - m)$.

\medskip

\noindent \textbf{Theorem} \ref{thm:cooperation} [Cooperating players explore more] (restated). \emph{Consider two cooperating players, Alice and Bob, playing a one armed bandit problem with discount factor $\beta \in (0,1)$. The left arm has success probability $p$ and the right arm has prior distribution $\mu$ that is not a point mass.
	Then there exists $\widehat{p} > g= g(\mu,\beta)$, so that for all $p < \widehat{p}$, at least one of the players explores the right arm with positive probability under any optimal strategy pair maximizing their total reward.}

\medskip

\begin{proof}
	In this proof we write $\Gamma_A = \Gamma_A(p)$ to emphasize the dependence on $p$.
	At $ p = g$, by definition of the Gittins index, Alice has a strategy starting at the right arm such that $\Ex(\Gamma_A(g)) = \frac{g}{1-\beta}\,.$
	Suppose Bob responds by playing L in rounds zero and one, then from round two onwards copies Alice's previous move. Then by inequality (\ref{eq:delta_gittins_myopic}), applied at $p=g$, we get
	$\Ex\Bigl(\Gamma_A(g) - \Gamma_B(g)\Bigr) \leq  (m-g) \beta\,.$
	Therefore
	$$
	\Ex\Bigl(\Gamma_A(g) + \Gamma_B(g)\Bigr) = 2 \Ex(\Gamma_A(g)) - \Ex\Bigl(\Gamma_A(g) - \Gamma_B(g)\Bigr) \geq \frac{2g}{1-\beta} + (g-m) \beta\,.
	$$
	Applying the same strategies at $p = g+ \delta$ and comparing the total rewards to both players staying left gives
	$$
	\Delta_p := \Ex\Bigl(\Gamma_A(g) + \Gamma_B(g)\Bigr) - {2p}/{(1-\beta)} = (g-m)\beta - {2\delta}/{(1-\beta)}  \,.
	$$
	By Lemma \ref{lem:gbound}, we get
	$
	\Delta_p \geq {\beta^2 w}/{2}  - {2\delta}/{(1-\beta)}\,.
	$
	This is positive when $\delta < {\beta^2 w (1-\beta)}/{4}\,.$
\end{proof}

Recall the index $g_T$ for finite horizon was defined in Remark~\ref{rmk:finite_index}.

\begin{theorem}[Cooperating players explore more, finite horizon] \label{thm:cooperation_finite_horizon}
	Consider two cooperating players, Alice and Bob, playing a one armed bandit problem. The left arm is known and has probability $p$, while the right arm has a known distribution $\mu$ that is not a point mass. Let $T$ be the horizon and $M^*$ the maximum of the support of $\mu$.
	Then there exists $\delta = \delta(T) > 0$ so that optimal players explore the right arm for all $p < g_T + \delta$.
\end{theorem}
\begin{proof}	
	Consider the same strategy for the players as in Theorem~\ref{thm:competing_explore_less_finite}:
	Alice plays the one player optimal, while
	Bob plays left in rounds zero and one, and then from round two onwards he copies Alice's move from the previous round.
	The expected total reward of Alice is
	$$
	\Ex(\Gamma_A) = m + \sum_{t = 1}^T \Ex(\gamma_A(t))\,.
	$$
	Bob's expected reward is
	$$
	\Ex(\Gamma_B) = 2p + \sum_{t=2}^T \Ex(\gamma_A(t-1)) = 2p + \sum_{t=1}^{T-1} \Ex(\gamma_A(t))\,.
	$$
	On the other hand, the expected total reward of the players if they never explore is $2p (T+1)$. The difference between the two strategies is
	$$
	\Delta = \Ex(\Gamma_A) + \Ex(\Gamma_B) - 2p(T+1) =2p + 2 \cdot \sum_{t=0}^T \Ex(\gamma_A(t)) - m - \Ex(\gamma_A(T)) - 2p(T+1)\,.
	$$
	Given that Alice is playing her single player optimal strategy, we have that for all $p \geq g_T$ her total reward is bounded as follows: $\sum_{t=0}^T \Ex(\gamma_A(t)) \geq g_T \cdot (T+1)$. Then
	\begin{align}
		\Delta & \geq 2p + 2g_T \cdot (T+1) - m - \Ex(\gamma_A(T)) - 2p(T+1) \notag \\
		& = 2g_T \cdot (T+1) - m - \Ex(\gamma_A(T)) - 2pT\,. \notag
	\end{align}
	Taking $p = g_T + \delta$ for $\delta \geq 0$, the previous inequality is equivalent to
	\begin{align}
		& \Delta \geq 2 g_T - m - \Ex(\Gamma_A(T)) - 2 \delta T \,. \notag
	\end{align}
	For $\Delta$ to be strictly positive it suffices that
	\begin{align} \label{eq:gamma_T_bound}
		\Ex(\gamma_A(T)) < 2g_T - m - 2\delta T
	\end{align}
	Recall that $\lim_{T \to \infty} g_T = M^*$ (Remark~\ref{rmk:finite_index}). Moreover, Alice's expected reward in the last round satisfies $ \Ex(\gamma_A(T)) \leq M^*$ when $p < M^*$.
	\medskip
	
	Let $\alpha = (M^* - m)/4$. Then there exists $T_0 = T_0(\alpha)$ so that for all $T \geq T_0$, we have $|g_T - M^*| < \alpha$. For all such $T$, by choice of $\alpha$ we have that
	$$
	2 g_T - m - \Ex(\gamma_A(T)) > 2 (M^* - \alpha) - m - M^* = \frac{M^* -m }{2} > 0
	$$
	Then inequality (\ref{eq:gamma_T_bound}) holds for all $\delta < (M^* - m)/(4T)$ as required.
\end{proof}

Note that cooperating players do prefer the predictable arm for high enough values of $p$. Recall that $M^*$ is the maximum of the support of $\mu$.

\begin{proposition}  \label{prop:cooperating_careful}
	There is a threshold $p^{\circ} < M^*$ so that cooperating players do not explore for any $p > p^{\circ}$.
\end{proposition}
\begin{proof}
	When $m < p$, the best case scenario for the cooperating players is that one of them (say Alice) plays right in round zero while Bob stays at left, and then from round one onwards both play right and the mean of the right arm is $M^*$ from round one onwards. Then any strategies $S_A$ and $S_B$ of the cooperating players will give at most this total reward, so:
	$$
	\Ex(\Gamma(S_A)) + \Ex(\Gamma(S_B)) \leq p + m + \frac{2\beta M^*}{1-\beta}
	$$
	The total reward of both players staying at the left arm is $2p/(1-\beta)$ which is more than any strategies involving exploration when
	\begin{align}
		p + m + \frac{2\beta M^*}{1-\beta} < \frac{2p}{1-\beta} \iff  p > \frac{(1-\beta)m +2 \beta M^*}{1+\beta} \notag
	\end{align}
	Setting $ p^{\circ} = \left((1-\beta)m + 2\beta M^*\right)/ (1+\beta)$ gives the required statement.
\end{proof}

\section{Neutral Play} \label{sec:neutral}

In this section we study the scenario of $\lambda = 0$, where the players are only interested in their own rewards. 

Recall the solution concepts are Nash equilibrium and perfect Bayesian equilibrium.
For neutral play the utility of each player is their total reward. The solution concepts will be Nash equilibrium and perfect Bayesian equilibrium.
Recall player $i$'s strategy $\sigma_i$ is a \emph{best response} to player $j$'s strategy $\sigma_j$ if no strategy $\sigma_i'$ achieves a higher expected utility against $\sigma_j$.
A mixed strategy profile $(\sigma_A, \sigma_B)$ is a \emph{Bayesian Nash equilibrium} if $\sigma_i$ is a best response for each player $i$. For brevity, we refer to such strategy profiles as Nash equilibria. 
A \emph{Perfect Bayesian Equilibrium} is the version of subgame perfect equilibrium for games with incomplete information. A pair of strategies $(\sigma_A, \sigma_B)$ is a perfect Bayesian equilibrium if $(i)$ starting from any information set, subsequent play is optimal, and $(ii)$ beliefs are updated consistently with Bayes' rule on every path of play that occurs with positive probability.
Such equilibria are guaranteed to exist in this setting (\cite{fudenberg}).




\begin{observation} \label{obs:neutral_simple} For all $p < g(\mu,\beta)$, in any Nash equilibrium, each player explores with strictly positive probability.
\end{observation}
\begin{proof}
	Suppose there is a Nash equilibrium with strategies $(S_A, S_B)$, in which one player - say Alice - explores with probability zero. Then $\Ex(\Gamma_A(S_A, S_B)) = p / (1 - \beta)$.
	
	Consider now the modified Alice strategy $S_A'$ of pulling the arm with the highest Gittins index in each round (ignoring any information from Bob). Then since $p < g(\mu,\beta)$, there is $\alpha \in (p, g)$ so that
	$\Ex(\Gamma_A(S_A', S_B)) = \alpha / (1 - \beta)
	$
	This is strictly higher than $p / (1-\beta)$, so $S_A'$ is an improving deviation, in contradiction with $(S_A, S_B)$ being an equilibrium. Then Alice explores with strictly positive probability.
\end{proof}

We will say that players learn from each other under some strategies if the total expected reward of each player is strictly higher than it would be for a single player using an optimal strategy. This can happen if the players infer additional information from each other's actions, beyond the bits that they observe themselves.

\medskip

\noindent \textbf{Theorem} \ref{neutral} [Neutral players learn from each other] (restated). \emph{Consider two neutral players, Alice and Bob, playing a one armed bandit problem with discount factor $\beta \in (0,1)$. The left arm has success probability $p$ and the right arm has prior distribution $\mu$ that is not a point mass. Then in any Nash equilibrium:}
	\begin{enumerate}
		\item \emph{For all $p < g(\mu,\beta)$, with probability $1$
		at least one player explores. Moreover, the probability that no player explores by time $t$ decays exponentially in $t$.}
		\item \emph{Suppose $p \in (p^*, g)$, where $p^*$ is the threshold above which competing players do not explore~\footnote{For the formal definition of $p^*$, see  Theorem~\ref{thm:competing_explore_less_discounted}.}. If the equilibrium is furthermore perfect Bayesian, then every (neutral) player has expected reward strictly higher than a single player using an optimal strategy.}
	\end{enumerate}
\begin{proof}
	Let $\alpha \in (p, g)$ denote the expected reward per round of a single player that follows the strategy of pulling the arm with the highest Gittins index in each round.
	
	\medskip
	
	\noindent \emph{\textbf{Part 1}}: Assume $p < g = g(\mu,\beta) $. Since $p < \alpha$, there exists $\epsilon > 0$ so that $p + \epsilon < \alpha$. Let $\Lambda = \alpha/(1-\beta)$ denote the expected total  reward achievable by a single optimal player.
	Fix a pair of strategies $(S_A, S_B)$ that define a Nash equilibrium. Let $\Phi_k = \Pr(D_k^c)$, where
	$$
	D_k = \Bigl\{\mbox{No player explores in rounds} \; \; \{0, \ldots, k-1\} \;\; \mbox{under strategies} \; \; (S_A, S_B)\,.\Bigr\}
	$$
	
	Since $(S_A, S_B)$ is an equilibrium, we can bound the expected reward of a single player by
	$$
	\Lambda \leq (1 - \Phi_k) \left[ \sum_{t = 0}^{k-1} p \cdot \beta^t + \sum_{t=k}^{\infty} 1 \cdot \beta^t \right]
	$$
	Rearranging the terms gives
	$$
	\Lambda \leq \sum_{t=0}^{k-1} \beta^t \Bigl[ (1-\Phi_k)p + \Phi_k \Bigr] + \sum_{t=k}^{\infty} \beta^t\,.
	$$
	Expanding the terms in the right hand side, we get a bound on $\alpha$:
	\begin{align}
		& \alpha \leq (1 - \beta^k) \Bigl[ p + \Phi_k(1-p) \Bigr] + \beta^k  \implies \frac{\alpha - p}{1-p} - \beta^k \leq \Phi_k\,. \notag
	\end{align}
	Choose $k$ so that $\beta^k \leq (\alpha - p)/2$. Then
	$$
	\Phi_k \geq \frac{\alpha - p}{2} \implies \Pr(D_k) \leq 1 - \frac{\alpha - p}{2}\,.
	$$
	The same argument gives that for every $\ell \geq 0$ we have
	\begin{align}
		\Pr(D_{k(\ell+1)} \; | \; D_{k\ell}) \leq 1 - (\alpha - p)/2 \notag
	\end{align}
	Thus inductively we obtain $\Pr(D_{k\ell}) \leq \Bigl( 1 - (\alpha - p)/2 \Bigr)^{\ell} $. This implies the required bound for $D_t$:
	\begin{align}
		& \Pr(D_t) \leq \left(1 - \frac{\alpha - p}{2} \right)^{\lfloor t/k \rfloor} \,. \notag
	\end{align}
	
	\noindent \emph{\textbf{Part 2}}: Assume that $p \in (p^*, g)$, where $p^*$ is the threshold from equation (\ref{def:p_star}) such that competing players do not explore for any $p > p^*$.
	In any perfect Bayesian equilibrium, the total expected reward of each player is at least $\alpha/(1-\beta)$, i.e. the single player optimum.
	
	Suppose towards a contradiction that there was a perfect Bayesian equilibrium with strategies $(S_A, S_B)$ where at least one of the players - say Bob - was getting exactly the expected reward of a single optimal player: $\Ex(\Gamma_B) = \alpha / (1 - \beta)$. Note that Alice must explore with positive probability regardless of what Bob does, since the single player optimum does better than never exploring.
	
	Consider now the following strategy $S_B'$ for Bob:
	stay left in every round until Alice explores for the first time.
	If Alice's exploration occurs in some round $k$, then from round $k+1$ onwards, let Bob use the optimal strategy that he would play when competing against Alice in the corresponding zero-sum game that starts at round $k$ (where Alice is forced to play right and Bob left in round $k$, and both play optimally afterwards). By Theorem~\ref{thm:competing_explore_less_discounted}, if Alice does explore the right arm at $k$, then when $p \in (p^*, g)$ Bob wins against Alice.
	
	Let $\Gamma_A'$ and $\Gamma_B'$ be the total rewards of the players
	under strategies $(S_A, S_B')$.
	Since the equilibrium is perfect Bayesian, we have $\Ex(\Gamma_A') \geq \alpha / (1 - \beta)$. Moreover, since Alice is not learning anything from Bob's strategy $S_B'$, she can only realize this minimum expected value by using an optimal strategy for a single player, which requires exploring with positive probability.
	\medskip
	
	We classify the trajectories realizable under strategies $(S_A, S_B')$ in two types, depending on whether Alice explores the right arm or not on that trajectory:
	\begin{enumerate}
		\item Alice does not explore: then Bob always plays left too, so  they get the same reward.
		\item Alice explores: then Bob's expected reward is strictly greater than Alice's total reward on that trajectory (since he wins in the corresponding zero sum game that starts with Alice's first exploration).
	\end{enumerate}
	Taking expectation over all possible trajectories and noting that Alice explores with positive probability, we get that Bob's deviation $S_B'$ ensures he gets expected total reward strictly greater than Alice, so
	$\Ex(\Gamma_B') > \Ex(\Gamma_A') \geq \alpha / (1 - \beta) = \Ex(\Gamma_B)\,.$
	Thus deviation $S_B'$ is profitable, in contradiction with $(S_A, S_B)$ being an equilibrium. Then both players get strictly more than the single player optimum.
\end{proof}

\begin{corollary}
	In each perfect Bayesian equilibrium, with positive probability the players do not have the same trajectory.
\end{corollary}
\begin{proof}
	The conclusion follows from part 2 of Theorem~\ref{neutral}. If the players always had the same trajectory, then the best total reward achievable would be the one player optimum.
\end{proof}

In contrast, there exist Nash equilibria where the players have the same trajectory with probability $1$. Such play can be supported by using (non-credible) threats in case one of the players deviates from the Nash equilibrium path, as it can be seen from the following example.

\begin{example}[Perfect Bayesian vs. Nash equilibrium] \label{eg:neutral_contrast}
	There is a Nash equilibrium in which both players follow the same trajectory and play the one player optimal strategy of pulling the arm with the highest Gittins index in each round.
	
	To see why this is the case, define the strategies of the players as follows:
	\begin{itemize}
		\item Pull the arm with the highest Gittins index in each round, breaking ties in favor of the left arm if both arms have the same index, as long as the other player did the same in the previous round.
		\item If in some round $k$ one of the players, say Alice, chooses a different action, then Bob switches to playing left forever from round $k+1$ onwards (and similarly if Bob deviates).
	\end{itemize}
	Note that with these strategies, if Alice deviates in some round $k$, then she cannot learn anything from Bob in rounds $\{k+1, k+2, \ldots\}$. Alice also did not learn anything from Bob in rounds $0$ to $k-1$ since they had the same trajectory and in round $k$ since she cannot see Bob's reward in round $k$ and his reward at $k$ does not affect his subsequent actions. Then Alice's expected reward overall cannot be better than the single player optimum, so the deviation is not strictly improving. Thus the strategies form a Nash equilibrium.
\end{example}


\begin{observation} \label{obs:neutral_play_left_at_g}
	For $p \geq g(\mu,\beta)$, there is a perfect Bayesian equilibrium where the players never explore: the pair of strategies in which each player stays left no matter what happens is a perfect Bayesian equilibrium. This is not a Nash equilibrium for $p < g(\mu, \beta)$, since a player can switch to the single player optimum.
\end{observation}

We show players also learn from each other in the finite horizon game. Recall $g_T$ is the index for finite horizon $T$ and $M^*$ is the maximum of the support of $\mu$. Formally, we have:

\begin{theorem}[Neutral players learn from each other, finite horizon]\label{thm:neutral_finite}
	Consider two neutral players, Alice and Bob, playing a one armed bandit problem with horizon $T$. The left arm has success probability $p$ and the right arm has prior distribution $\mu$ that is not a point mass. Then in any subgame perfect equilibrium:
	\begin{enumerate}
		\item For all $p < M^*$, with probability that converges to  $1$ as $T$ grows,
		at least one player explores. Moreover, the probability that no player explores by time $t$ decays polynomially in $t$.
		\item Suppose $(m + M^*)/2 < p < M^* $. Then every (neutral) player has expected reward strictly higher than a single optimal player.
	\end{enumerate}
\end{theorem}
\begin{proof}
	Recall $g_T$ denotes the index of the right arm.
	Let $\alpha \cdot (T+1)$ be the expected reward of an optimal single player, where $\alpha \in (p, g_T)$.
	
	\medskip
	
	\noindent \textbf{\emph{Part 1}}: Let $p < M^*$. We show  there exist $\psi = \psi(\mu, p) > 0$ and $C < \infty$
	so that the inequality
	$$
	\Pr\Bigl(\mbox{No player explores in rounds } \{0, \ldots, T\} \mbox{ under }(S_A, S_B)\Bigr) \leq C \cdot T^{-\psi}
	$$
	holds for all Nash equilibria $(S_A, S_B)$.
	
	First, since $p < M^*$, there exist $T_1$ and $\alpha > p$ so that for all $T \geq T_1$, the single player optimum over rounds $\{0, \ldots, T\}$ is in expectation at least $\alpha \cdot (T+1)$.
	
	For a sequence $\{T_j\}$ monotonically decreasing (to be specified later) and for $T>T_1$, write 
	$
	D_j = \Bigl\{ \mbox{No player explores in rounds } \{0, \ldots, T\} \mbox{ under } (S_A, S_B) \Bigr\}
	$.
	\medskip
	
	\begin{figure}[H]
		\centering
		\includegraphics[scale=0.9]{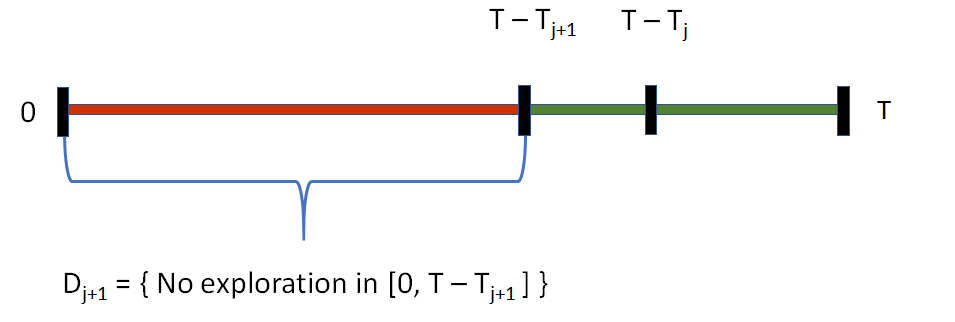}
		\caption{Depiction of the intervals induced by the sequence  ${T_j}$ and the events $D_j$.}
		\label{fig:intervals_finite}
	\end{figure}
	
	\medskip
	
	Fix $j$ and let $k = T_{j+1} - T_j$. A depiction of the intervals induced by the sequence $T_j$ is given in Figure \ref{fig:intervals_finite}. Write $\Phi_j = \Pr(D_j^c \; | \; D_{j+1})$. Since $(S_A, S_B)$ is a Nash equilibrium,
	\begin{align}
		& \alpha \cdot (T_j + 1) \leq (1-\Phi_j) \cdot \Bigl[ p_k + T_j + 1\Bigr] + \Phi_j \cdot \Bigl[ T_{j+1} + 1 \Bigr] \notag \\
		& \alpha \cdot (k + T_j + 1)   \leq T_j + 1 + \Phi_j \cdot (1-p) \cdot k + p \cdot k \,.\notag
	\end{align}
	This implies the inequality:
	$$
	\Phi_j \cdot (1-p) \cdot k \geq (\alpha - p)k - (1-\alpha)(T_j+1)\,.
	$$
	Choose $k$ such that $(\alpha - p) \cdot k \in \Bigl[ T_j+1,T_j+2 \Bigr]$. We obtain
	\begin{align} \label{eq:delta_tj}
		0 \leq T_{j+1} - \left( 1 + \frac{1}{\alpha - p}\right) T_j \leq \frac{2}{\alpha - p}
	\end{align}
	Then
	$$
	\alpha \cdot T_j \leq k \cdot \Phi_j \leq \Phi_j \cdot \left( \frac{T_j + 2}{\alpha - p} \right) \implies \Phi_j \geq \frac{\alpha (\alpha - p) T_j}{T_j + 2} \geq \frac{\alpha(\alpha-p)}{2}\,.
	$$
	Now suppose $T \in (T_{\ell}, T_{\ell+1})$. Then
	\begin{align} \label{eq:bound_dl}
		\Pr(D_{\ell-1}) \leq \Pr(D_{\ell-1} | D_{\ell}) \leq 1 - \frac{\alpha(\alpha-p)}{2} \implies \Pr(D_{\ell}) \leq \left(1 - \frac{\alpha(\alpha-p)}{2} \right)^{\ell-1}
	\end{align}
	Since $\{T_j\}$ grows exponentially, by (\ref{eq:delta_tj}), the claim follows from (\ref{eq:bound_dl}).
	\medskip
	
	\noindent \textbf{\emph{Part 2}}: The proof is similar to the discounted case, so we give a sketch highlighting the differences.
	Let $(S_A, S_B)$ be a perfect Bayesian equilibrium where Bob gets exactly the single player optimum: $\Ex(\Gamma_B) = \alpha  \cdot (T+1)$.
	Since the horizon $T$ is large, Alice explores with positive probability (see part 2 of Theorem~\ref{thm:competing_explore_less_finite}).
	
	Consider deviation $S_B'$ for Bob: stay left in every round until Alice explores for the first time (e.g. in round $k$), and from round $k+1$ onwards, use the optimal zero-sum strategy.
	Let $\Gamma_A'$ and $\Gamma_B'$ be the total rewards
	under $(S_A, S_B')$.
	Since the equilibrium is perfect Bayesian, $\Ex(\Gamma_A') \geq \alpha \cdot (T+1)$. Alice is not learning from Bob under $S_B'$, so she must explore with positive probability.
	We classify the trajectories realizable under $(S_A, S_B')$ in two types:
	\begin{enumerate}
		\item Alice does not explore: then Bob always plays left too, so  they get the same total reward.
		\item Alice explores: then by  Theorem~\ref{thm:competing_explore_less_finite}, Bob gets strictly more than Alice for $p > (m + M^*)/2$.
	\end{enumerate}
	Taking expectation over all possible trajectories and noting that Alice explores with positive probability, it follows that Bob's deviation $S_B'$ ensures he gets expected total reward strictly greater than Alice.
	Then we get
	$\Ex(\Gamma_B') > \Ex(\Gamma_A') \geq \alpha \cdot (T+1) = \Ex(\Gamma_B)\,.$
	Then Bob's deviation $S_B'$ is profitable, in contradiction with the choice of equilibrium strategies $(S_A, S_B)$.
\end{proof}

\section{Long Term Behavior} \label{app:long_term}
We show that in every Nash equilibrium, competing and neutral players converge to playing the same arm forever from some point on. For $\lambda \not \in [-1,0]$, this can fail.

\subsection{Nash Equilibria with Oscillations} 
\label{subsec:oscillations_nash}

In this section we construct Nash equilibria with oscillations when $\lambda > 0$ and $\lambda < -1$.
The next example shows that for $\lambda > 0$ there cannot be a general theorem that the players settle on the same arm in every Nash equilibrium. The construction shows that for every $\lambda > 0$ and every sufficiently large discount factor $\beta < 1$, there is a Nash equilibrium in which one of the players alternates between the two arms infinitely often.

\begin{proposition}[Nash equilibria where players do not converge, $\lambda > 0$] \label{eg:alpha_positive}
	Suppose Alice and Bob are playing a one-armed bandit problem with discount factor $\beta$, where the left arm has success probability $p$ and the right arm has prior distribution $\mu$ that is a point mass at $m > p$.
	
	Then for every $\lambda > 0$, for any discount factor $\beta > 1/(1+\lambda)$, there is a Nash equilibrium in which Bob visits both arms infinitely often.
\end{proposition}
\begin{proof}
	Let $k \in \mathbb{N}$. Define strategies $(S_A, S_B)$ by:
	\begin{enumerate}
		\item Bob's strategy $S_B$ is to visit the left arm in rounds $0, k, 2k, 3k, \ldots$, and the right arm in the remaining rounds, no matter what Alice does.
		\item Alice's strategy $S_A$ is to stay at the right arm if Bob follows the trajectory above. If Bob ever deviates from $S_B$ for the first time in some round $\ell$, then Alice plays left forever starting with round $\ell+1$.
	\end{enumerate}
	To show the strategies are in Nash equilibrium, consider first the total rewards $\Gamma_A$ and $\Gamma_B$ obtained by Alice and Bob, respectively, on the main path under strategy pair $(S_A, S_B)$ starting from round zero:
	$$
	\Gamma_A  = \frac{m}{1-\beta} \;\; \mbox{  and  } \;\; \Gamma_B = \frac{m}{1-\beta} - \frac{m - p}{1-\beta^k}
	$$
	It is clear that Alice has no incentive to deviate, since her deviation does not change Bob's behavior and adding any round of playing the left arm only worsens her own total rewards.
	Consider next any strategy $S_B' $ of Bob that deviates from $S_B$ in some round $\ell$. Observe first that the best case deviation for Bob is when he deviates in some round $\ell$ divisible by $k$ so that instead of playing left in the next round he plays right.
	Then w.l.o.g. $\ell=0$. The expected total rewards under $(S_A, S_B')$ starting from round zero are
	$$
	\Gamma_A' = m + \beta \cdot \frac{p}{1-\beta} \; \; \mbox{ and } \; \; \Gamma_B' \leq \frac{m}{1-\beta}
	$$
	Then Bob's utility under strategy pairs $(S_A, S_B)$ and $(S_A, S_B')$, respectively, is
	$$
	u_B = \frac{(1 + \lambda) \cdot m}{1-\beta} - \frac{m - p}{1-\beta^k} \; \; \mbox{ and } \; \; u_B'  \leq \frac{m}{1-\beta} + \lambda \left( m + \beta \cdot \frac{p}{1-\beta} \right)
	$$
	If $\beta > 1 / (1 + \lambda)$ and $k$ is large enough, then
	$1 - \beta < \beta \cdot \lambda \cdot (1 - \beta^k)$. This implies that
	$u_B'  < u_B$, so $(S_A, S_B)$ is indeed a Nash equilibrium.
\end{proof}
\begin{figure}[h!]
	\centering
	\includegraphics[scale = 0.8]{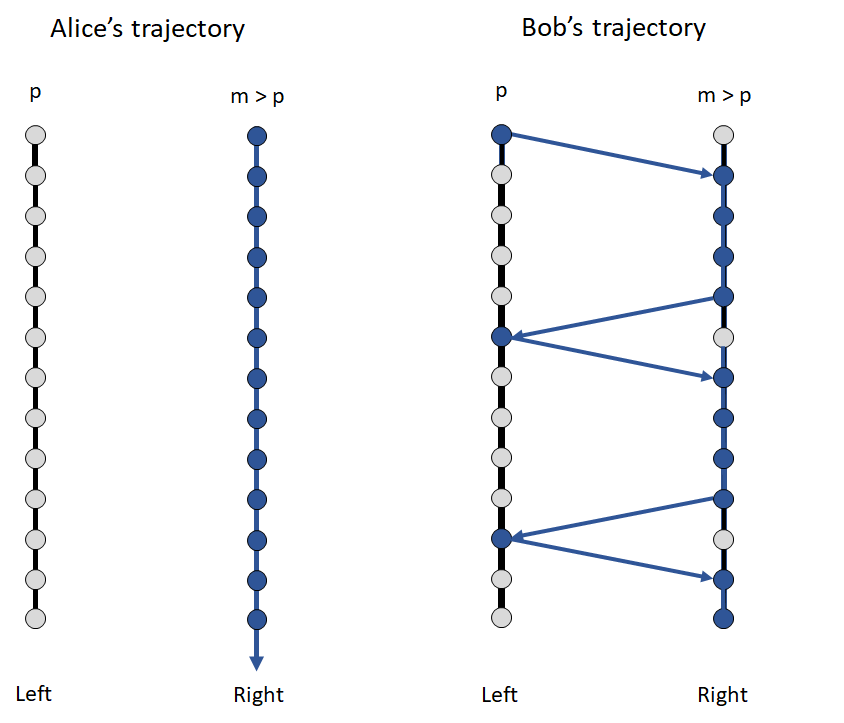}
	\caption{Trajectories of the players on the main line under strategies $(S_A, S_B)$ for $\lambda > 0$ and $k=5$. The circles represent time units $0, 1, 2, \ldots$. The left arm has success probability $p$ and the right arm has prior distribution which is a point mass at $m > p$.}
	\label{fig:oscillating_cooperating}
\end{figure}

Note this Nash equilibrium is not Pareto optimal: both players could improve their utility by a joint deviation to always playing right. This Nash equilibrium is also not subgame perfect because it involves a non-credible threat by Alice (of playing left forever after any deviation by Bob).

\begin{remark}
	In the example discussed above, if $\beta < 1/(1+\lambda)$, then the only Nash equilibrium is where both players always play right.
\end{remark}

Proposition~\ref{eg:alpha_positive} does not preclude the possibility that the players settle on the same arm in every perfect Bayesian equilibrium when $\lambda > 0$.


For every $\lambda < -1$, there are perfect Bayesian equilibria where the players do not settle on the same arm in the long term (for some choice of $\mu$ and $\beta$), and where a player oscillates infinitely often between the arms.

\begin{proposition}[Perfect equilibria where competing players do not converge, $\lambda < -1$]  \label{prop:nash_oscillating_competing}
	Suppose Alice and Bob are playing a one-armed bandit problem with discount factor $\beta$, where the left arm has success probability $p$ and the right arm has prior distribution $\mu$ which is a point mass at $m > p$. For every $\lambda < -1$, if the discount factor satisfies $\beta^2 \cdot |\lambda| > 1$, then
	there is a perfect Bayesian equilibrium in which Bob visits both arms infinitely often.
\end{proposition}
Note in this case there is no uncertainty since both arms are known, so perfect Bayesian equilibria coincide with subgame perfect equilibria.\\

\begin{proof}[Proof of Proposition \ref{prop:nash_oscillating_competing}]
	Suppose Alice fixes a trajectory $\tau$ for Bob that requires him to play right in rounds $0, k, 2k, 3k, \ldots$, and to play left in all other rounds. Let $S_A$ be the Alice strategy of playing left as long as Bob follows trajectory $\tau$, while if Bob ever deviates from $\tau$ in some round $s$, then $S_A$ switches to playing right forever from round $s+1$ onwards. Let $S_B$ be the Bob strategy of following  trajectory $\tau$, unless Alice ever deviates from playing left in some round $s$; in that case, $S_B$ switches to playing right forever from round $s+1$ onwards. Whenever the players find themselves off the main line given by this prescription, they just play right forever.
	\begin{figure}[h!]
		\centering
		\includegraphics[scale = 0.9]{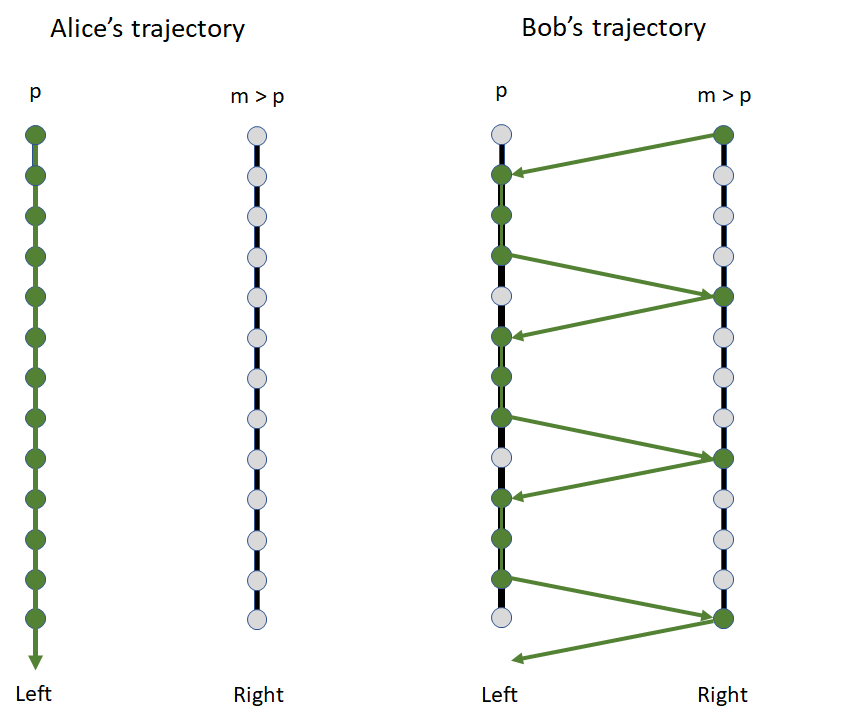}
		\caption{Trajectories of the players on the main line under strategies $(S_A, S_B)$ for $\lambda < -1$ and $k=4$. The circles represent time units $0, 1, 2, \ldots$.}
		\label{fig:oscillating_competing}
	\end{figure}
	
	We claim that $(S_A, S_B)$ represent a subgame perfect equilibrium. 
	The expected total rewards of the players under $(S_A, S_B)$ on the main line of play are:
	\begin{align} \label{eq:rewards_main_line_oscill_competing}
		\Ex(\Gamma_A(S_A,S_B)) = \frac{p}{1-\beta} \; \mbox{ and } \; \Ex(\Gamma_B(S_A, S_B)) = \frac{m-p}{1-\beta^k} + \frac{p}{1-\beta}
	\end{align}
	Then the expected utilities are:
	\begin{align} \label{eq:utils_main_line_oscill_competing}
		\Ex(u_A(S_A,S_B)) &= (1 + \lambda) \cdot \frac{p}{1-\beta} + \lambda \cdot \frac{m-p}{1-\beta^k} \notag \\
		\Ex(u_B(S_A, S_B)) &= (1 + \lambda) \cdot \frac{p}{1-\beta} + \frac{m-p}{1-\beta^k}
	\end{align}

	To check the equilibrium property, it will suffice to compare the utility each player $i$ gets when following $S_i$ with the utility obtained when deviating to playing left forever or to playing right forever. Note the response of the other player to any deviation is to switch to playing right forever, thus if there is an improving deviation $S_i'$ for player $i$, then there is an improving deviation $S_i''$ in which $i$ plays a fixed arm forever. Moreover, since the right arm is always better, that deviation will be to play right forever.

	If Bob switches to playing right from some round on, then the highest gain can be obtained when the switch takes place from round $1$ onwards. In this case, Alice will play right forever from round $2$. Thus
	\begin{align} \label{eq:rewards_bobright_oscill_competing}
		\Ex(\Gamma_A(S_A, Right)) = p + \beta p + \beta^2 \cdot \frac{m}{1-\beta} \; \mbox{ and } \; \Ex(\Gamma_B(S_A, Right)) = \frac{m}{1-\beta}
	\end{align}
	Bob's expected utility is
	\begin{align} \label{eq:utils_bobright_oscill_competing}
		\Ex(u_B(S_A, Right)) = (1+ \lambda) \cdot \frac{m}{1-\beta} + \lambda \cdot [(p - m)(1+\beta)]
	\end{align}
	From equations (\ref{eq:utils_bobright_oscill_competing}) and (\ref{eq:utils_main_line_oscill_competing}), we get that $\Ex(u_B(S_A, Right)) < \Ex(u_B(S_A, S_B))$ whenever
	\begin{align}
		1 + \beta + \ldots + \beta^{k-1} > \frac{1 + \lambda(1 + \beta)}{1 + \lambda} = \frac{|\lambda| \cdot (1+ \beta)-1}{|\lambda|-1}
	\end{align}
	Such a $k$ exists whenever $|\lambda| \beta > 1$.
	
	For Alice the best time to deviate is just before $S_B$ tells Bob to play right (in round $k-1$). 
	\begin{align} \label{eq:utils_alice_rightnonzero_oscill_competing_deviation}
		\Ex(u_A'(k-1, \infty)) = \frac{m}{1-\beta} + \lambda \cdot \left[ p + \beta \cdot \frac{m}{1-\beta} \right]
	\end{align}
	On the other hand, under strategies $(S_A, S_B)$, Alice's utility from round $i$ on is
	\begin{align}  \label{eq:utils_alice_rightnonzero_oscill_competing}
		\Ex(u_A(k-1, \infty)) = \frac{p}{1-\beta} + \lambda  p + \lambda \beta \cdot \left(\frac{m-p}{1-\beta^k} + \frac{p}{1-\beta} \right)
	\end{align}
	Comparing Alice's utility before and after the deviation, we obtain
	\begin{align}
		\Ex(u_A'(k-1, \infty)) &< \Ex(u_A(k-1, \infty)) \notag \\
		\frac{m-p}{1-\beta} + \frac{\lambda \beta (m-p)}{1-\beta} & < \frac{\lambda \beta(m-p)}{1-\beta^k}  \iff \notag \\
		1 + \lambda \beta & < \lambda \beta \cdot \frac{1 - \beta}{1-\beta^k} = \frac{\lambda \beta}{1 + \beta + \ldots + \beta^{k-1}} \iff \notag \\
		1 + \beta + \ldots + \beta^{k-1} &> \frac{|\lambda| \beta}{|\lambda| \beta - 1}
	\end{align}
	This holds whenever $|\lambda| \beta^2 > 1$ and $k$ is large enough.
	%
\end{proof}

\subsection{A Concentration Lemma} \label{app:concentration_lemma}

\begin{lemma} \label{lem:epsilon_bound}
	Suppose Alice is playing a one armed bandit problem, where the left arm has success probability $p$ and the right arm has prior distribution $\mu$ that is not a point mass. Let $R_{\infty}$ denote Alice's total number of pulls of the right arm, $R_n$ the number of pulls by time $n$, and $w_n$ the number of successes in the first $\min\{n, R_{\infty}\}$ draws by Alice.
	
	Then for each $\epsilon > 0$ and $k \in \mathbb{N}$, we have
	\begin{align}  \label{eq:expectation_psi_k}
		\Pr\Bigl(\bigl(|w_k - k \Theta| > k \epsilon  \bigr) \cap \bigl(R_{\infty} \geq k \bigr) \Bigr) \leq 2 e^{-2k\epsilon^2}
	\end{align}
\end{lemma}
\begin{proof}
	Recall we are working in the probability space where $\Theta$ is picked according to $\mu$ and every toss of the risky arm is a Bernoulli variable with parameter $\Theta$.
	
	The goal is to show the next inequality holds for every $k$:
	\begin{align} \label{eq:goal_r_infinity}
		\Pr\Bigl( \bigl( |w_k - k \Theta| > k \epsilon \bigr)  \cap \bigl( R_{\infty} \geq k \bigr) \mid \Theta \Bigr) \leq 2 e^{-2k\epsilon^2}
	\end{align}
	
	Fix a value of $k$.
	It will be enough to show that the following inequality holds for every $n$:
	\begin{align} \label{eq:goal_r_n}
		\Pr\Bigl( \bigl( |w_k - k \Theta| > k \epsilon \bigr)  \cap \bigl( R_{n} \geq k \bigr) \mid \Theta \Bigr) \leq 2 e^{-2k\epsilon^2}
	\end{align}
	We define an auxiliary variable $\widetilde{w}_k$ to represent the number of successes in the first $k$ pulls, but where the pulls are padded in case their number is too low. That is, if $R_n \geq k$, then $\widetilde{w}_k = w_k$. Otherwise, if $R_n < k$, then add another $k - R_n$ pulls after time $n$ and let $\widetilde{w}_k$ denote the number of successes in the first $k$ pulls defined this way.
	Note the containment
	$$
	\Bigl( \bigl( |w_k - k \Theta| > k \epsilon \bigr)  \cap \bigl( R_{n} \geq k \bigr) \Bigr)  \subseteq \Bigl( |\widetilde{w}_k - k \cdot \Theta| > k \epsilon \Bigr).
	$$
	Using Hoeffding's inequality applied to the $k$ pulls of the risky arm, we have
	\begin{align} \label{eq:tilde_w}
		\Pr\Bigl( |\widetilde{w}_k - k \cdot \Theta| > k \epsilon \mid \Theta\Bigr) \leq 2 e^{-2k \epsilon^2}
	\end{align}
	This implies inequality (\ref{eq:goal_r_n}), and so  (\ref{eq:goal_r_infinity}) also holds. Taking expectation in (\ref{eq:goal_r_infinity}) implies the required inequality (\ref{eq:expectation_psi_k}).
\end{proof}

\subsection{Neutral Players in the Long Term} \label{app:neutral_players}

\begin{theorem}[Neutral players eventually settle on the same arm] \label{thm:neutral_settle_same_arm}
	Consider two neutral players, Alice and Bob, playing a one armed bandit problem with discount factor $\beta \in (0,1)$. The left arm has success probability $p$ and the right arm has prior distribution $\mu$ such that $\mu(p)=0$. Then in any Nash equilibrium, with probability $1$ the players eventually settle on the same arm.
\end{theorem}

\begin{proof}[Proof of Theorem~\ref{thm:neutral_settle_same_arm}]
	Let $(S_A, S_B)$ be an arbitrary pair of strategies in Nash equilibrium.
	For each round $n$, let $H_n^i$ be the history from player $i$'s point of view up to round $n$; this contains the actions of both players until the end of round $n-1$, the rewards of player $i$ in the first $n-1$ rounds, and the action prescribed to player $i$ by strategy $S_i$ for round $n$.
	Let $N$ be a stopping time for player $i$.
	
	Let $\widehat{\Theta}_k^i$ be the empirical mean of the success probability at the right arm as observed by player $i$ after $k$ explorations, defined if $R_{\infty}^i \geq k$, where $R_{\infty}^i$ is player i's total number of pulls of the right arm.
	
	Let $\epsilon \in (0, 1/2)$. Pick $\delta  \in (0, \epsilon)$ such that $\mu([p-3\delta, p+3\delta]) < \epsilon$.
	Let $\delta_1 \in (0, \delta)$ be such that $\mu([p-3\delta_1,p+3\delta_1]) \leq   \epsilon^2 \delta / 2$. 
	
	We start by showing that if Alice explores at least $k$ times, for $k$ large enough, then she will settle afterwards on the better of the two arms with high probability. Then we will study Bob's behavior and show that he will also pick the better of the two arms if at least one of the players explores $k$ times.
	
	We bound Alice's expected value for pulling the right arm in the future by considering three cases, depending on the value of $\widehat{\Theta}_k^A$.
	
	\medskip
	
	\noindent \emph{\textbf{Case 1:}} $\widehat{\Theta}_k^A \leq p - 2\delta$. The goal will be to show that the following event has small probability
	$$\widetilde{D}_1 = \{(R_{\infty}^A \geq k+1) \cap (\widehat{\Theta}_k^A \leq p - 2 \delta)\}.$$
	Let $N$ be the (random) time of the $k+1$-st exploration by Alice, if it exists, and otherwise $N = \infty$.
	Define the random variable
	\begin{align} \label{def:psi_k}
		\Psi_k^A = \Pr\Bigl((|w_k - k \Theta| > k \delta) \cap ( R_{\infty}^A \geq k+1) \mid H_{N}^A \Bigr)
	\end{align}
	
	Using Lemma~\ref{lem:epsilon_bound}, we get $\Ex(\Psi_k^A) \leq 2 e^{-2k\delta^2}$. From this we deduce by Markov's inequality that
	\begin{align} \label{eq:psi_k_chance}
		\Pr(\Psi_k^A \geq e^{-k\delta^2}) \leq 2 e^{-k\delta^2}
	\end{align}
	We define the event $D_1 = \{\Psi_k^A \leq e^{-k\delta^2}\} \cap \widetilde{D}_1$. By Lemma~\ref{lem:epsilon_bound}, we have
	\begin{align} \label{eq:bound_D1_tilde_minus_D}
		\Pr(\widetilde{D}_1 \setminus D_1) \leq 2e^{-k\delta^2}
	\end{align}
	If $|\widehat{\Theta}_k^A - \Theta| \leq \delta$ and $D_1$ holds, then $\Theta \leq p - \delta$.
	Let $\Gamma_A(N,\infty) = \sum_{j=0}^{\infty} \beta^j \cdot \gamma_A(N+j)$ denote Alice's total normalized reward starting from round $N$ under strategy $S_A$.
	Then Alice's total reward satisfies:
	\begin{small}
		\begin{align}
			\mathbbm{1}_{D_1} \Ex\Bigl[ \Gamma_A(N, \infty) \mid H_{N}^A , \Theta \Bigr] \leq \mathbbm{1}_{D_1} \Bigl[ \mathbbm{1}_{|\Theta - \widehat{\Theta}_k^A| \leq \delta} \left( p - \delta + \frac{p \beta}{1-\beta} \right) + \mathbbm{1}_{|\Theta - \widehat{\Theta}_k^A| > \delta} \frac{1}{1-\beta} \Bigr]
		\end{align}
	\end{small}
	
	Taking expectation over $\Theta$ gives (since $D_1 \subseteq \{R_{\infty}^A \geq k+1\}$):
	\begin{small}
		\begin{align} \label{eq:bound_exp_SA_N}
			\mathbbm{1}_{D_1} \Ex\Bigl[ \Gamma_A(N,\infty) \mid H_{N}^A \Bigr] & \leq \mathbbm{1}_{D_1} \Bigl[p - \delta + \frac{p \beta}{1-\beta} + \Pr\Bigl(\{R_{\infty}^A \geq k+1\} \cap |\Theta - \widehat{\Theta}_k^A| > \delta \mid H_{N}^A \Bigr) \cdot \frac{1}{1-\beta} \Bigr] \notag \\
			& = \mathbbm{1}_{D_1} \Bigl[ \frac{p}{1-\beta} - \delta + \Psi_k^A  \cdot \frac{1}{1-\beta} \Bigr] \notag \\
			& \leq \mathbbm{1}_{D_1} \Bigl[ \frac{p}{1-\beta} - \delta + \frac{e^{-k\delta^2}}{1-\beta} \Bigr]
		\end{align}
	\end{small}
	Consider an alternative strategy $S_A'$ as follows: play $S_A$ until time $N$ when it is about to do its $k+1$-st exploration. If $\widehat{\Theta}_k^A \leq p - 2 \delta$ and $\Psi_k^A \leq e^{-k\delta^2}$, then play left forever, otherwise continue with $S_A$. Note that $S_A'$ differs from $S_A$ only on the event $D_1$.
	
	Using inequality (\ref{eq:bound_exp_SA_N}), on the event $D_1$ we have that for all $k > \delta^{-2} \cdot \log{1/(\delta(1-\beta))}$:
	\begin{align} \label{eq:bound_case1_star}
		\Ex\Bigl[ \Gamma_A(N,\infty) - \Gamma_A'(N,\infty) \mid H_{N}^A \Bigr] \leq \Bigl( \frac{e^{-k\delta^2}}{1-\beta} - \delta \Bigr) < 0
	\end{align}
	
	Since the original strategies $(S_A, S_B)$ are in equilibrium, we have
	\begin{align}
		0 \leq \Ex(\Gamma_A) - \Ex(\Gamma_A') & = \Ex\Bigl[\beta^N \Gamma_A(N,\infty) - \beta^N \Gamma_A'(N,\infty) \Bigr] \notag \\
		& = \Ex\Bigl[ \beta^N \mathbbm{1}_{D_1} \Ex\Bigl( \Gamma_A(N,\infty) - \Gamma_{A'}(N,\infty) \mid H_{N}^A \Bigr) \Bigr]
	\end{align}
	
	This forces $\Pr(D_1) = 0$ by (\ref{eq:bound_case1_star}).
	By inequality (\ref{eq:bound_D1_tilde_minus_D}), we get that $\Pr(\widetilde{D}_1) \leq 2e^{-k\delta^2}$.
	
	\bigskip
	
	\noindent \emph{\textbf{Case 2:}} $\widehat{\Theta}_k^A \geq p + 2\delta$.
	Define $M$ as the first round after the $k$-th exploration at which Alice does not explore. If there is no such time, then $M = \infty$. The goal will be to bound the probability of the event
	\begin{align}
		\widetilde{D}_2 = \{(M < \infty) \cap (\widehat{\Theta}_k^A \geq p + 2\delta)\}
	\end{align}
	
	Recall ${H}_M^A$ is the history from Alice's point of view until the beginning of round $M$ and $\Phi_k^A$ as the random variable
	\begin{align} \label{def:phi_k}
		\Phi_k^A = \Pr\Bigl((|w_k - k \Theta| > k \epsilon) \cap (  M < \infty ) \mid {H}_{M}^A \Bigr)
	\end{align}
	
	Using Lemma~\ref{lem:epsilon_bound}, we get $\Ex(\Phi_k^A) \leq 2 e^{-2k\epsilon^2}$. From this we deduce by Markov's inequality that
	\begin{align}
		\Pr(\Phi_k^A \geq e^{-k\epsilon^2}) \leq 2 e^{-k\epsilon^2}
	\end{align}
	
	Define the event $D_2 = \{\Phi_k^A \leq e^{-k \delta^2} \} \cap \widetilde{D}_2$. By Lemma~\ref{lem:epsilon_bound}, we have $$\Pr(\widetilde{D}_2 \setminus D_2) \leq 2e^{-k\delta^2}\,.$$
	
	Consider Alice's strategy $S_A$. If $M$ is finite, then at round $M$ Alice is playing left by definition. We will upper bound Alice's total reward from round $M$ onwards and compare it to the expected reward obtained by playing right in all rounds $t \geq M$. The following inequality is immediate:
	$\max\{\Theta,p\} \leq \Theta + p \cdot \mathbbm{1}_{\Theta < p}\,.$
	By taking expectation and conditioning on the history, on the event $D_2$ we get
	$\Ex(\mathbbm{1}_{\Theta < p} \mid {H}_{M}^A) \leq e^{-k\delta^2}.$
	Then on the event $D_2$, Alice's expected total reward from round $M$ onwards under strategy $S_A$ satisfies
	\begin{align} \label{eq:gamma_total_case2_ub}
		\Ex\Bigl[\Gamma_A(M,\infty) \mid {H}_M^A\Bigr] & \leq p + \frac{\beta}{1-\beta} \cdot \Ex(\Theta \mid {H}_M^A) + \frac{\beta \cdot p}{1-\beta} \cdot \Pr\Bigl((\Theta < p) \cap (M < \infty) \mid {H}_M^A\Bigr) \notag \\
		& \leq p + \frac{\beta}{1-\beta} \cdot \Ex(\Theta \mid {H}_M^A) + \frac{\beta \cdot p}{1-\beta} \cdot \Phi_k^A
	\end{align}
	
	On the other hand, on the event $D_2$, by playing right forever starting with round $M$, Alice's expected reward is $\Gamma_A'$, which has expectation:
	\begin{align} \label{eq:gamma_prime_case2_ub}
		\Ex\Bigl[\Gamma_A'(M,\infty) \mid {H}_{M}^A\Bigr] = \frac{ \Ex(\Theta \mid {H}_M^A)}{1-\beta}
	\end{align}
	
	If $\widehat{\Theta}_k^A - \Theta \leq \delta$ and the event $D_2$ holds, then $\Theta \geq p + \delta$.
	On the event $D_2$ we have
	\begin{align} \label{eq:gamma_oneround_case2_lb}
		\Ex\Bigl[\Theta \mid {H}_M^A\Bigr] \geq p + \delta - e^{-k\delta^2}
	\end{align}
	By combining inequalities (\ref{eq:gamma_total_case2_ub}), (\ref{eq:gamma_prime_case2_ub}), and (\ref{eq:gamma_oneround_case2_lb}), we get that on the event $D_2$ the following inequalities hold for all $k$ large enough so that $e^{-k\delta^2} < \delta(1-\beta)$.
	\begin{small}
		\begin{align} \label{eq:delta_right_forever_neutral}
			\Ex\Bigl[\Gamma_A(M, \infty) \mid {H}_{M}^A\Bigr] - \Ex\Bigl[\Gamma_A'(M,\infty) \mid {H}_{M}^A\Bigr] & \leq p + \frac{\beta \cdot p}{1-\beta} \cdot \Phi_k^A + \frac{\beta}{1-\beta} \cdot \Ex(\Theta \mid {H}_M^A) - \frac{ \Ex(\Theta \mid {H}_M^A)}{1-\beta} \notag \\
			& \leq p + \frac{\beta p}{1-\beta} \cdot e^{-k\delta^2} - \Ex(\Theta \mid {H}_M^A) \notag \\
			& \leq p + \frac{\beta}{1-\beta} \cdot e^{-k\delta^2} - p - \delta + e^{-k\delta^2} \notag \\
			& = \frac{e^{-k\delta^2}}{1-\beta} - \delta < 0
		\end{align}
	\end{small}
	Consider the alternative strategy $S_A'$ defined as follows: play $S_A$ until time $M$ when it is about to play left. If $\Phi_k^A \leq e^{-k\delta^2}$ and $\widehat{\Theta}_k^A \geq p + 2 \delta$, then play right forever. Otherwise, continue with $S_A$. Similarly to Case 1, since $(S_A, S_B)$ is an equilibrium, we obtain that $\Ex(\Gamma_A) - \Ex(\Gamma_A') \geq 0$. Moreover, the strategies $S_A$ and $S_A'$ have different rewards only on the event $D_2$, so
	\begin{align}
		0 \leq \Ex(\Gamma_A) - \Ex(\Gamma_A') & = \Ex\Bigl[\Gamma_A(M,\infty) - \Gamma_A'(M,\infty)\Bigr] \notag \\
		& = \Ex\Bigl[\beta^M \mathbbm{1}_{D_2}  \Ex\Bigl(\Gamma_A(M,\infty) - \Gamma_A'(M,\infty)  \mid {H}_{M}^A \Bigr) \Bigr]
	\end{align}
	
	By (\ref{eq:delta_right_forever_neutral}), we get $\Pr(D_2) = 0$. Recall $\Pr(\widetilde{D}_2 \setminus D_2) \leq 2e^{-k\delta^2}$. Then
	$\Pr(\widetilde{D}_2) \leq 2e^{-k\delta^2}$.
	
	\bigskip
	
	\noindent \emph{\textbf{Case 3:}} $\Theta \in [p - 3 \delta, p + 3 \delta]$.
	By choice of $\delta$, we have $\mu([p - 3 \delta, p + 3 \delta]) < \epsilon$.
	
	\bigskip
	
	\noindent \emph{\textbf{Case 4:}} $|\Theta - \widehat{\Theta}_k^A| > \delta$ and $R_{\infty}^A \geq k$. By Lemma~\ref{lem:epsilon_bound}, this happens with probability at most $2 \epsilon^2$.
	
	\bigskip
	
	Let $\tau_k^A$ be the round in which Alice explores for the $k$-th time (if $R_{\infty}^A < k$, then $\tau_k^A = \infty$). From cases $(1-4)$, we obtain that for $k$ additionally satisfying the inequality $e^{-k\delta^2} < \epsilon$ we have
	\begin{align} \label{eq:pr_alice_wrong}
		\Pr\Bigl(\exists \; t > \tau_k^A \; : \; \mbox{ Alice pulls the worse arm at time } t\Bigr) \leq 7\epsilon
	\end{align}
	
	We consider now Bob's behavior and show that if Alice explores at least $k+1$ times, then Bob will explore at all later times with high probability.
	Let $Q > \tau_{k+1}^A$ be the first round after $\tau_{k+1}^A$ where $S_B$ plays left; if there is no such round, then $Q = \infty$.
	
	Let $\widetilde{H}_{Q}$ be Alice's public history (i.e. containing the sequence of arms she played, but not her rewards) running from round $0$ until the end of round $Q-1$. If $Q = \infty$, then $\widetilde{H}_Q$ is the whole history. Note that $\widetilde{H}_Q$ is observable by Bob. 
	
	Let $\mu_{\widetilde{H}_Q}$ denote Bob's posterior distribution for $\Theta$ (at the end of round $Q - 1$) given $\widetilde{H}_Q$:
	\begin{align} \label{eq:posterior_bob}
		{\mu}_{\widetilde{H}_Q}(a,b) =   \Pr\Bigl(\Theta \in (a,b) \mid \widetilde{H}_Q\Bigr) \cdot \mathbbm{1}_{Q < \infty}
	\end{align}
	
	For $Q < \infty$, the history $\widetilde{H}_Q$ is said to be ``good'' if $\mu_{\widetilde{H}_Q}\bigl(p,p+ 2\delta\bigr) < \sqrt{\epsilon}$ and
	$\mu_{\widetilde{H}_Q}\bigl(0,p\bigr) < \epsilon \delta$.
	Define strategy $S_B'$ as follows: play $S_B$ until the beginning of round $Q$, when $S_B$ plays left. If the history is good, then play right forever. Otherwise, continue with $S_B$. The goal will be to compare the total reward $\Gamma_B$ from strategy $S_B$ with the reward $\Gamma_B'$ from $S_B'$.
	
	Taking expectation over the history in (\ref{eq:posterior_bob}) gives
	$$\Ex\bigl[\mu_{\widetilde{H}_Q}\bigl(a,b\bigr)\bigr] = \Pr\bigl(\{\Theta \in (a,b) \} \cap \{Q < \infty\}\bigr)\,.$$
	We additionally require that
	$$\frac{6e^{-k\delta_1^2}}{1-\beta} \leq \epsilon^2 \delta /2\,.$$
	By choice of $\delta$, $\delta_1$, and $k$, we get
	\begin{align} \label{eq:exp_bound_barmu}
		& \bullet \; \; \;  \Ex\bigl[\mu_{\widetilde{H}_Q}\bigl(p,p+2\delta\bigr)\bigr]  \leq \mu\bigl(p,p+2\delta\bigr) \leq \epsilon \notag \\
		& \bullet \; \; \; \Ex\bigl[\mu_{\widetilde{H}_Q}\bigl(0,p\bigr)\bigr]  \leq \Pr\bigl(\{\Theta \in [0,p]\} \cap \{R_{\infty}^A \geq k+1\}\bigr)
		\leq \mu(p-3\delta_1, p) + \frac{6e^{-k\delta_1^2}}{1-\beta} \leq \epsilon^2 \delta
	\end{align}
	Applying Markov's inequality in (\ref{eq:exp_bound_barmu}) gives:
	$
	\Pr\bigl( \mu_{\widetilde{H}_Q}\bigl(p,p+2\delta\bigr) \geq \sqrt{\epsilon}, \; Q < \infty \bigr) \leq \sqrt{\epsilon}
	$
	and
	$
	\Pr\bigl(\mu_{\widetilde{H}_Q}\bigl(0,p\bigr) \geq \epsilon \delta, \; Q < \infty \bigr) \leq \epsilon
	$. Combining these implies:
	\begin{align} \label{eq:bad_history_finite_Q_neutral}
		\Pr\bigl((Q < \infty) \cap (\widetilde{H}_Q \mbox{ is bad})\bigr)
		\leq 2 \sqrt{\epsilon}\,.
	\end{align}
	
	\medskip
	
	We claim that $\Pr\bigl(Q < \infty \mbox{ and } \widetilde{H}_Q \mbox{ is good}\bigr) = 0$. On the event that $\widetilde{H}_Q$ is good, we have
	\begin{align} \label{eq:exp_gamma_B}
		\Ex\Bigl[\Gamma_B(Q, \infty) \mid \widetilde{H}_Q \Bigr] & \leq p + \frac{\beta}{1-\beta} \Ex\Bigl[ \max(\Theta,p) \mid \widetilde{H}_Q \Bigr] \notag \\& \leq p + \frac{\beta}{1-\beta} \Bigl( \Ex\bigl[\Theta \mid \widetilde{H}_Q\bigr] + p \cdot \Pr(\Theta < p \mid \widetilde{H}_Q) \Bigr) \notag \\
		& \leq p + \frac{\beta}{1-\beta} \Ex\bigl[\Theta \mid \widetilde{H}_Q \bigr] + \epsilon \delta
	\end{align}
	On the other hand, the reward from strategy $S_B'$ on the event that $\widetilde{H}_Q$ is good is
	\begin{align}\label{eq:exp_gamma_B_prime}\Ex\Bigl[\Gamma_B'(Q, \infty) \mid \widetilde{H}_Q\Bigr] = \frac{\Ex\bigl[\Theta \mid \widetilde{H}_Q\bigr]}{1-\beta}\,.
	\end{align}
	Using inequalities (\ref{eq:exp_gamma_B}) and (\ref{eq:exp_gamma_B_prime}), the difference between the rewards under $S_B$ and $S_B'$ from round $Q$ onwards given a good history $\widetilde{H}_Q$ satisfies the inequality:
	\begin{align} \label{eq:delta_b}
		\Delta(\widetilde{H}_Q) = \Ex\Bigl[\Gamma_B(Q, \infty) - \Gamma_B'(Q,\infty) \mid \widetilde{H}_Q\Bigr] \leq p - \Ex\Bigl[\Theta \mid \widetilde{H}_Q\Bigr] + \epsilon \delta
	\end{align}
	For any good history $\widetilde{H}_Q$, we can bound $\Ex[\Theta \mid \widetilde{H}_Q]$ as follows:
	\begin{align}
		\Ex\Bigl[\Theta \mid \widetilde{H}_Q\Bigr] & \geq (p+2\delta) \cdot  \mu_{\widetilde{H}_Q}\bigl(p+2\delta,1\bigr) + p  \cdot \mu_{\widetilde{H}_Q}\bigl(p,p+2\delta\bigr) \notag \\
		& = p \cdot \mu_{\widetilde{H}_Q}\bigl(p,1\bigr) + 2 \delta \cdot \mu_{\widetilde{H}_Q}\bigl(p+2\delta,1\bigr)
	\end{align}
	Then on the event that $\widetilde{H}_Q$ is good:
	$
	\Ex[\Theta \mid \widetilde{H}_Q] \geq p (1-\epsilon \delta) + \delta \geq p +  \delta/2
	$. Replacing in (\ref{eq:delta_b}) implies that
	\begin{align} \label{eq:delta_hq_bound}
		\Delta(\widetilde{H}_Q) = \Ex\Bigl[\Gamma_B(Q, \infty) - \Gamma_B'(Q,\infty) \mid \widetilde{H}_Q\Bigr] < -\delta/3
	\end{align}
	for all good histories $\widetilde{H}_Q$ when $\epsilon < 1/2$.
	Then the difference between the rewards under strategies $S_B$ and $S_B'$ satisfies
	\begin{align} \label{eq:delta_b_goodhq_neutral}
		0 \leq \Ex(\Gamma_B - \Gamma_B') = \Ex\Bigl[\beta^Q \cdot \Delta(\widetilde{H}_Q) \cdot \mathbbm{1}_{\{\widetilde{H}_Q \; \mbox{is} \; \mbox{good}\}}\Bigr] \,.
	\end{align}
	This forces $\Pr\bigl((\widetilde{H}_Q \; \mbox{is good}) \cap (Q < \infty)\bigr) = 0$. Combining (\ref{eq:bad_history_finite_Q_neutral}) and (\ref{eq:delta_b_goodhq_neutral}) implies
	\begin{align} \label{eq:q_finite_ub}
		\Pr(Q < \infty) \leq 2 \sqrt{\epsilon}
	\end{align}
	
	By (\ref{eq:q_finite_ub}), the probability that Alice explores infinitely often and Bob plays left infinitely often is bounded by $2 \sqrt{\epsilon}$ for every $\epsilon$, so the probability is zero. Similarly with roles reversed.
	
	Now we can conclude the proof of convergence to the same arm. If both players explore finitely many times, then they converge to the left arm. Otherwise, one of the players - say Alice - explores infinitely often. This implies that Bob plays right for all but finitely many rounds. Reversing the roles, Alice also plays right for all but finitely many rounds.
\end{proof}

We give an example showing why the $\Psi_k^A$ condition is needed. This example has the property that the empirical mean $\Theta_k^A$ is significantly below $p$, yet the posterior mean is larger than $p$.

\begin{example}
	Let $p = 0.7$ and $\mu$ be the uniform distribution on $[0,1/4] \cup [3/4,1]$. Suppose that in some history Alice obtains after $k$ observations a value of $\widehat{\Theta}_k^A = 0.65 < p$. Then since $\mu$ has zero mass in $[1/4,3/4]$, it is much more likely that $\Theta \in [3/4,1]$ rather than $\Theta \in [0,1/4]$. In this case, $\Psi_k^A$ is large and the strategy $S_A'$ does not tell Alice to stop exploring, even though $\widehat{\Theta}_k^A < p - 2 \delta$ for small $\delta > 0$.
\end{example}

\subsection{Competing Players in the Long Term} \label{app:settle_competing}

\begin{theorem}[Competing players eventually settle on the same arm] \label{thm:competing_settle_same_arm}
	Consider two competing players, Alice and Bob, playing a one armed bandit problem with discount factor $\beta \in (0,1)$. The left arm has success probability $p$ and the right arm has prior distribution $\mu$ with $\mu(p)=0$. Then in any Nash equilibrium, with probability $1$ the players eventually settle on the same arm.
\end{theorem}
In this setting there is an additional difficulty compared to the neutral case: players might refrain from switching to the optimal arm in order to not trigger an adverse reaction from the opponent. The key to overcoming this difficulty is realizing that it is impossible for both players to benefit from repeatedly pulling the inferior arm. Each player might subjectively believe for some time that they are playing the better arm, but if a player keeps exploring, then their subjective evaluation of the risky arm will converge to the objective reality. \\

\begin{proof}[Proof of Theorem~\ref{thm:competing_settle_same_arm}]
	Let $(S_A, S_B)$ be an optimal strategy pair. Note the strategies may be randomized.
	
	As in the neutral case, for each $n$, let $H_n^i$ be the history from player $i$'s point of view up to round $n$.
	Let $N$ be a stopping time for player $i$; this means that for every fixed $n$, the event $N = n$ is determined by $H_n^i$.
	Then $H_N^i$ is a history for player $i$ up to round $N$.
	
	Recall $\widehat{\Theta}_k^i$ is the fraction of successes in the first $k$ explorations by player $i$ and $R_{\infty}^i$ is the total number of explorations by player $i$.
	
	\medskip
	
	Define the events
	\begin{itemize}
		\item $\Lambda_f = \{\mbox{Alice explores infinitely often and Bob explores finitely many times}\}$
		\item $\Lambda_{\infty} = \{\mbox{Both players explore infinitely often}\}.$
	\end{itemize}
	
	\medskip
	
	Choose the following parameters:
	\begin{align} \label{eq:epsilon_bound_competing}
		\bullet &  \; \; \; \epsilon \in \Bigl(0, \frac{1-\beta}{2}\Bigr) \; \mbox{ and } \notag \\
		\bullet &  \; \; \;  \delta \in (0, \epsilon) \; \mbox{ such that } \; \mu([p-10\delta, p + 10\delta]) < \epsilon \; \mbox{ and } \mu(p - 8\delta) = 0
	\end{align}
	
	\noindent \emph{\textbf{Case 1:}} $\Theta < p$. We consider two sub-cases, depending on whether both players explore infinitely often or only one player does (note that on the event that both players explore finitely often the conclusion holds immediately).
	
	\medskip
	
	\emph{\textbf{Case 1.a:}} $\Theta < p$, Alice explores infinitely often and Bob finitely many times (i.e., $\Lambda_f$ holds).
	
	Define the event
	$S_k = \{\mbox{Bob explores at or after Alice's k-th exploration}\}$.
	Let $$\epsilon_1 = \delta(1-\beta)/4\,.$$
	
	Let $N_{\ell}$ be the time of Alice's $\ell$-th exploration if it exists, and otherwise $N_{\ell} = \infty$.
	Let $\mathcal{B}_1(k)$ be the collection of histories $H_{N_k}^{A}$ such that $\Pr(\Lambda_f \cap S_k \mid H_{N_k}^A) > \epsilon_1$ and $\mathcal{B}_2(k)$ the collection of histories $H_{N_k}^A$ for which $\Pr(\Lambda_f^c \mid H_{N_k}^A) > \epsilon_1$. Applying the Levy zero-one law gives
	\begin{align} \label{eq:A_cap_HK_B2}
		\lim_{k \to \infty} \Pr\Bigl(\Lambda_f \cap \bigl(H_{N_k}^A \in \mathcal{B}_2(k)\bigr)\Bigr) = 0\,.
	\end{align}
	We also have $\bigcap_{k=0}^{\infty} (\Lambda_f \cap S_k) = \emptyset$ by the definition of $\Lambda_f$ and $S_k$. Together with (\ref{eq:A_cap_HK_B2}), this implies that for large enough $k$ the next two conditions are satisfied:
	\begin{align} \label{eq:def_k_competing}
		\Pr(\Lambda_f \cap S_k) \leq \delta \epsilon_1 \; \mbox{ and } \Pr\Bigl(\Lambda_f \cap \bigl(H_{N_k}^A \in \mathcal{B}_2(k)\bigr)\Bigr) \leq \epsilon_1\,.
	\end{align}
	Fix $k$ so that (\ref{eq:def_k_competing}) is satisfied for all $k' \geq k$.
	Letting $k' = k+1$ and taking expectation over the history gives $$\Pr(\Lambda_f \cap S_{k+1}) = \Ex(\Pr(\Lambda_f \cap S_{k+1}) \mid H_{N_{k+1}}^A) \leq \delta \epsilon_1\,.$$
	
	The collection of \emph{bad} histories is defined as $\mathcal{B} = \mathcal{B}_1(k+1) \cup \mathcal{B}_2(k+1)$. Note that if the history $H_{N_{k+1}}^A$ is good, then
	\begin{align} \label{eq:upper_bound_s_k+1}
		\Pr(S_{k+1} \mid H_{N_{k+1}}^A) = \Pr(\Lambda_f^c \mid H_{N_{k+1}}^A) + \Pr(\Lambda_f \cap S_k \mid H_{N_{k+1}}^A) \leq 2 \epsilon_1
	\end{align}
	
	\smallskip
	
	By Markov's inequality, we obtain:
	\begin{align} \label{eq:pr_bad_history}
		\Pr(H_{N_{k+1}}^A \in \mathcal{B}_1(k+1)) & = \Pr\Bigl(\Pr(\Lambda_f \cap S_{k+1} \mid H_{N_{k+1}}^A) > \epsilon_1\Bigr) \notag \\
		& \leq \frac{\Ex(\Pr(\Lambda_f \cap S_{k+1} \mid H_{N_{k+1}}^A))}{\epsilon_1}
		\leq \frac{\delta \epsilon_1}{\epsilon_1} = \delta
	\end{align}

	By (\ref{eq:pr_bad_history}) and (\ref{eq:def_k_competing}), we have
	\begin{align} \label{eq:A_cap_H_k_bad}
		\Pr(\Lambda_f \cap (H_{N_{k+1}}^A \mbox{ is bad})) \leq \epsilon_1 + \delta
	\end{align}
	
	\medskip
	
	Then we can define, similarly to Theorem~\ref{thm:neutral_settle_same_arm}, the event $\widetilde{D}_1$ and the random variable $\Psi_k^A$ (both from Alice's point of view):
	\begin{align}
		\widetilde{D}_1 & = \{(R_{\infty}^A \geq k+1) \cap (\widehat{\Theta}_k^A \leq p - 2 \delta)\} \notag \\
		\Psi_k^A & = \Pr\Bigl((|w_k - k \Theta| > k \delta) \cap ( R_{\infty}^A \geq k+1) \mid H_{N_{k+1}}^A \Bigr) \notag
	\end{align}
	Define the event $$D_1 = \{\Psi_k^A \leq e^{-k\delta^2}\} \cap \widetilde{D}_1 \cap \{H_{N_{k+1}}^A \mbox{ is good}\}\,.$$
	
	Using inequality (\ref{eq:bound_exp_SA_N}) in Theorem \ref{thm:neutral_settle_same_arm}, we get that Alice's total reward from round $N_{k+1}$ on is bounded as follows:
	\begin{align} \label{eq:bound_exp_SA_N_competing}
		\mathbbm{1}_{D_1} \Ex\Bigl[ \Gamma_A(N_{k+1},\infty) \mid H_{N_{k+1}}^A \Bigr] & \leq \mathbbm{1}_{D_1} \Bigl[ \frac{p}{1-\beta} - \delta + \frac{e^{-k\delta^2}}{1-\beta} \Bigr]
	\end{align}
	Bob's expected total reward given Alice's information is bounded, using (\ref{eq:upper_bound_s_k+1}), by
	\begin{align} \label{eq:bound_exp_SB_N_competing_nonexploringbob}
		\mathbbm{1}_{D_1} \Ex\Bigl[ \Gamma_B(N_{k+1},\infty) \mid H_{N_{k+1}}^A \Bigr] & \geq \mathbbm{1}_{D_1} \Bigl[ \frac{p}{1-\beta} \cdot (1- 2\epsilon_1) \Bigr]
	\end{align}
	Combining inequalities (\ref{eq:bound_exp_SA_N_competing}) and (\ref{eq:bound_exp_SB_N_competing_nonexploringbob}) implies that Alice's expected utility satisfies
	\begin{align} \label{eq:bound_exp_uA_N_competing}
		\mathbbm{1}_{D_1} \Ex\Bigl[ u_A(N_{k+1},\infty) \mid H_{N_{k+1}}^A \Bigr] & = \mathbbm{1}_{D_1} \Ex\Bigl[ \Gamma_A(N_{k+1},\infty) - \Gamma_B(N_{k+1},\infty) \mid H_{N_{k+1}}^A \Bigr] \notag \\
		& \leq \mathbbm{1}_{D_1} \Bigl[ \frac{p}{1-\beta} - \delta + \frac{e^{-k\delta^2}}{1-\beta} - \frac{p}{1-\beta} \cdot (1- \delta(1-\beta)/2) \Bigr] \notag \\
		& = \mathbbm{1}_{D_1} \Bigl[ \frac{p \delta}{2} - \delta + \frac{e^{-k\delta^2}}{1-\beta} \Bigr]
	\end{align}
	We require $k$ to be large enough so that
	\begin{align} \label{eq:competing_1a_bound_k}
		\frac{e^{-k\delta^2}}{1-\beta} \leq \delta \epsilon
	\end{align}
	Consider instead the following strategy $S_A'$ for Alice: play $S_A$ until time $N_{k+1}$ when it is about to do its $k+1$-st exploration. If
	$\widehat{\Theta}_k^A \leq p - 2\delta$ and $\Psi_k^A \leq e^{-k\delta^2}$ and $\{H_{N_{k+1}}^A \mbox{ is good}\}$,
	then play left forever. (In other words, Alice deviates from $S_A$ exactly on the event $D_1$.)
	Otherwise, continue with $S_A$.
	Then Alice's expected utility from round $N_{k+1}$ on under strategy pair $(S_A', S_B)$ satisfies
	\begin{align} \label{eq:bound_exp_uA_prime_N_competing}
		\mathbbm{1}_{D_1} \Ex\Bigl[ u_A'(N_{k+1},\infty) \mid H_{N_{k+1}}^A \Bigr] \geq \mathbbm{1}_{D_1} \Bigl[ \frac{p}{1-\beta} - \Bigl( \frac{p}{1-\beta} + \frac{e^{-k\delta^2}}{1-\beta} \Bigr) \Bigr] = \mathbbm{1}_{D_1} \Bigl[  \frac{-e^{-k\delta^2}}{1-\beta} \Bigr]
	\end{align}
	
	From inequalities (\ref{eq:bound_exp_uA_N_competing}), (\ref{eq:competing_1a_bound_k}) (\ref{eq:bound_exp_uA_prime_N_competing}), we get that on the event $D_1$, for $\epsilon < 1/4$, the inequality $u_A'(N_{k+1}, \infty) > u_A(N_{k+1}, \infty)$ holds. Since the original strategy pair $(S_A, S_B)$ is optimal, we get that $\Pr(D_1) = 0$.
	
	Now consider the event $\bigl( \Lambda_f \cap (\Theta < p) \bigr) \setminus D_1$ and note that it is contained in the union of the following events:
	\begin{itemize}
		\item $\Theta \in [p-3\delta, p]$. Note that $\mu([p-3\delta, p]) < \epsilon$.
		\item $\Lambda_f \cap (|\widehat{\Theta}_k^A - \Theta| > \delta)$. By Lemma 4, the probability of this event is bounded by $2 e^{-2k\delta^2} \leq 2 \delta \epsilon$.
		\item $\Psi_k^A > e^{-k\delta^2}$. By inequality (\ref{eq:psi_k_chance}), the probability of this event is at most $2e^{-k\delta^2} \leq 2 \delta \epsilon$.
		\item $\Lambda_f \cap (H_{N_{k+1}}^A \mbox{ is bad})$. By (\ref{eq:A_cap_H_k_bad}), the probability of this event is at most $\epsilon_1 + \delta$.
	\end{itemize}
	
	Since $\Pr(D_1) = 0$, we conclude that
	$$
	\Pr\Bigl(\Lambda_f \cap (\Theta < p)\Bigr) = \Pr\Bigl(\bigl( \Lambda_f \cap (\Theta < p) \bigr) \setminus D_1\Bigr) \leq 7 \epsilon
	$$
	Thus taking $\epsilon \to 0$ gives $\Pr(\Lambda_f \cap (\Theta < p)) = 0$, so case (1.a) occurs with probability zero.
	
	\medskip
	
	\emph{\textbf{Case 1.b:}} $\Theta < p$, Both players explore infinitely often (i.e., $\Lambda_{\infty}$ holds).
	
	\medskip
	
	For each random time $N \in \mathbb{N}$ at which player $i$ has explored at least $k$ times, define
	$$
	\Psi_k^i = \Pr\Bigl((|\widehat{\Theta}_k^i - \Theta| > \delta) \cap (R_{\infty}^i \geq k+1) \mid H_{N}^i \Bigr)
	$$
	
	For each player $i$ and each $\ell \in \mathbb{N}$, let $\tau_{\ell}^i$ be the (random) time of the $\ell$-th exploration by that player.
	Let $N$ be the first time strictly greater than $\max\{\tau_k^A, \tau_k^B\}$ at which Alice explores. If there is no such time, then $N = \infty$.
	
	Define the event
	$$D = \{\widehat{\Theta}_k^A < p - 8\delta\} \cap \{\Psi_k^A \leq \delta \} \cap \{\max{(\tau_k^A, \tau_k^B)} < N < \infty\}.$$
	Let $u_{N}^A(S_A, S_B) = \Gamma_A(N, \infty) - \Gamma_B(N, \infty)$ be Alice's normalized expected utility from round $N$ onwards under $(S_A, S_B)$ and define $u_{N}^A(Left, S_B)$ similarly when she plays  left forever from round $N$ on the event $D$.
	Since strategy pair $(S_A, S_B)$ is optimal, Alice does at least as well by playing $S_A$ as she would do by playing left from round $N$ onwards on the event $D$:
	\begin{align} \label{eq:uA_bound_above_competing_1b}
		\Ex\Bigl[u_N^A(S_A, S_B) \mathbbm{1}_D \mid H_N^A\Bigr] & \geq \Ex\Bigl[u_N^A(Left, S_B) \mathbbm{1}_D \mid H_N^A\Bigr] \notag \\
		& \geq \Bigl( p - \Ex\Bigl[\gamma_B(N) \mid H_N^A\Bigr] - \frac{\beta \cdot e^{-k\delta^2}}{1-\beta} \Bigr) \mathbbm{1}_D,
	\end{align}
	since Bob's action at time $N$ is not influenced by Alice's action at time $N$.
	On the other hand, Alice's utility can be bounded from above by
	\begin{small}
		\begin{align}
			\Ex\Bigl[u_N^A(S_A, S_B) \mathbbm{1}_D \mid H_N^A\Bigr] & \leq  \Ex\Bigl[\gamma_A(N) \mathbbm{1}_D \mid H_N^A\Bigr] - \Ex\Bigl[\gamma_B(N)\mathbbm{1}_D \mid H_N^A\Bigr] + \Ex\Bigl[u_{N+1}^A(S_A, S_B) \mathbbm{1}_D \mid H_N^A\Bigr] \notag
		\end{align}
	\end{small}
	Alice's expected reward in round $N$ on the event $D$ is at most
	\begin{align}
		\Ex\Bigl[\gamma_A(N) \mathbbm{1}_D \mid H_N^A\Bigr] & \leq \Ex\Bigl[\Theta \cdot \mathbbm{1}_D \mid H_N^A\Bigr]\notag \\
		& = \Ex\Bigl[\Theta \cdot \mathbbm{1}_D \mathbbm{1}_{|\widehat{\Theta}_k^A - \Theta| > \delta} \mid H_N^A\Bigr] + \Ex\Bigl[\Theta \cdot \mathbbm{1}_D \mathbbm{1}_{|\widehat{\Theta}_k^A - \Theta| \leq \delta} \mid H_N^A\Bigr]  \notag \\
		& \leq \Psi_k^A \cdot \mathbbm{1}_D + \Ex\bigl[ (p-7 \delta)\mathbbm{1}_D \mid H_N^A \bigr] \notag \\
		& \leq (p - 6\delta) \mathbbm{1}_D
	\end{align}
	Then we can bound Alice's normalized expected utility from round $N$ on by
	\begin{small}
		\begin{align} \label{eq:uA_bound_below_competing_1b}
			\Ex\Bigl[u_N^A(S_A, S_B) \mathbbm{1}_D \mid H_N^A\Bigr] & \leq  (p - 6\delta) \mathbbm{1}_D -  \Ex\Bigl[\gamma_B(N) \mathbbm{1}_D \mid H_N^A\Bigr] +  \Ex\Bigl[u_{N+1}^A(S_A, S_B) \mathbbm{1}_D \mid H_N^A\Bigr]
		\end{align}
	\end{small}
	Let $\eta \in (0, \delta)$ so that $\mu([p-8\delta-2\eta, p-8\delta+2\eta]) < \delta(1-\beta)\epsilon$.
	Select $k$ large enough so that
	\begin{align} \label{eq:condition_k_competing_1b}
		\frac{e^{-k\eta^2}}{1-\beta} \leq \delta \epsilon
	\end{align}
	Comparing inequality (\ref{eq:uA_bound_below_competing_1b}) to (\ref{eq:uA_bound_above_competing_1b}) yields
	\begin{align} \label{eq:alice_final_lower_bound_competing_1b}
		\Ex\Bigl[u_{N+1}^A(S_A, S_B) \mathbbm{1}_D \mid H_N^A\Bigr] \geq \Bigl(6\delta - \frac{\beta \cdot e^{-k \delta}}{1-\beta}\Bigr) \mathbbm{1}_D \geq 5\delta \cdot \mathbbm{1}_D
	\end{align}
	Taking expectation over the history gives
	\begin{align} \label{eq:alice_final_D_lower_bound_competing_1b}
		\Ex\Bigl[u_{N+1}^A(S_A, S_B) \mathbbm{1}_D\Bigl] \geq 5\delta \cdot \Pr(D)
	\end{align}
	Define the event $$D^* = \Bigl\{\Theta < p - 8\delta\Bigr\} \cap \Bigl\{\max{(\tau_k^A, \tau_k^B)} < N < \infty\Bigr\}\,.$$
	We argue that the event $D^*$ is well approximated by the event $D$ by showing their symmetric difference is small. By the choice of $k$ in (\ref{eq:condition_k_competing_1b}), we have
	$$\Pr(\Psi_k^A \geq \delta) \leq 2\delta(1-\beta)\epsilon \; \mbox{ and } \; \Pr(|\widehat{\Theta}_k^A - \Theta| > \eta) \leq 2\delta(1-\beta)\epsilon\,.$$
	The symmetric difference is included in the union of the following events:
	$$
	D^* \triangle D \subset  \Bigl\{\Theta \in \bigl[p-8\delta-2\eta, p-8\delta+2\eta\bigr] \Bigr\} \cup \Bigl\{|\widehat{\Theta}_k^A - \Theta| > \eta \Bigr\} \cup \Bigl\{\Psi_k^A \geq \delta \Bigr\}
	$$
	By choice of $\delta, \eta$, and $k$, we obtain
	$$\Pr(D^* \triangle D) = \Ex\bigl(|\mathbbm{1}_D - \mathbbm{1}_{D^*}| \bigr) \leq 5 \delta(1-\beta)\epsilon
	$$
	This implies that
	\begin{align} \label{eq:symmetric_difference_u_bound_1b_competing}
		\Ex\Bigl[u_{N+1}^A(S_A, S_B) |\mathbbm{1}_{D^*} - \mathbbm{1}_{D}| \Bigl] \leq \frac{1}{1-\beta} \cdot 5 \delta(1-\beta)\epsilon = 5 \delta \epsilon
	\end{align}
	Combining (\ref{eq:symmetric_difference_u_bound_1b_competing}) with (\ref{eq:alice_final_D_lower_bound_competing_1b}) implies
	\begin{align} \label{eq:lower_bound_u_alice_1b_competing}
		\Ex\Bigl[u_{N+1}^A(S_A, S_B) \mathbbm{1}_{D^*}\Bigl] \geq 5\delta \cdot \Pr(D) - 5 \delta \epsilon
	\end{align}
	Define an event similar to $D$ for Bob:
	$$
	D' = \Bigl\{\bigl(\widehat{\Theta}_k^B < p - 8 \delta\bigr) \cap
	\bigl( \Psi_k^B \leq \delta \bigr) \cap \bigl(\max\{\tau_k^A, \tau_k^B\} < N \bigr)
	\Bigr\}
	$$
	Then Bob's normalized expected utility from round $N$ onwards given his history satisfies
	\begin{align}
		\Ex\Bigl[ u_{N+1}^B(S_A, S_B) \cdot \mathbbm{1}_{D'} \mid H_N^B \Bigr] & \geq \Ex\Bigl[ u_{N+1}^B(S_A, Left) \cdot \mathbbm{1}_{D'} \mid H_N^B\Bigr]  \\
		& \geq -\frac{e^{-k\delta^2}}{1-\beta} \mathbbm{1}_{D'} \geq - \delta \epsilon
	\end{align}
	A similar argument to the proof of (\ref{eq:symmetric_difference_u_bound_1b_competing}) gives
	\begin{align} \label{eq:lower_bound_u_bob_1b_competing}
		\Ex\Bigl[u_{N+1}^B(S_A, S_B) \mathbbm{1}_{D*}\Bigr] \geq -\delta \epsilon - 5 \delta \epsilon
	\end{align}
	Since the game is zero-sum, adding (\ref{eq:lower_bound_u_alice_1b_competing}) and (\ref{eq:lower_bound_u_bob_1b_competing}) gives
	$$
	0 = \Ex\Bigl[u_{N+1}^A(S_A, S_B) \mathbbm{1}_{D*}\Bigr] + \Ex\Bigl[u_{N+1}^B(S_A, S_B) \mathbbm{1}_{D*}\Bigr] \geq 5\delta \cdot \Pr(D) - 11 \delta \epsilon
	$$
	This implies that $\Pr(D) \leq 3 \epsilon$.
	\medskip
	
	Note the containment
	$$
	\Bigl\{ \bigl(\Theta < p - 9 \delta\bigr) \cap \bigl(\max\{\tau_k^A, \tau_k^B\}< N < \infty \bigr) \Bigr\} \subset D \cup \Bigl\{\bigl(|\Theta - \widehat{\Theta}_k^A| \geq \delta\bigr) \cup \bigl(\Psi_k^A > \delta \bigr)\Bigr\}
	$$
	Since $\Pr(D) \leq 3 \epsilon$ and $\Pr(|\Theta - \widehat{\Theta}_k^A| \geq \delta) \leq \epsilon$ and $\Pr(\Psi_k^A > \delta) \leq \epsilon$, we get
	$$
	\Pr\Bigl(\bigl(\Theta < p - 9 \delta\bigr) \cap \bigl(\max\{\tau_k^A, \tau_k^B\}< N < \infty \bigr) \Bigr) \leq 5 \epsilon \,.
	$$
	By the choice of $\delta$, since $\bigl(R_{\infty}^A = R_{\infty}^B = \infty\bigr) \subseteq \bigl(\max\{\tau_k^A, \tau_k^B\}< N < \infty \bigr)$, this implies that
	$$
	\Pr\Bigl( \bigl(\Theta < p\bigr) \cap \bigl(R_{\infty}^A = R_{\infty}^B = \infty\bigr) \Bigr) \leq 6 \epsilon
	$$
	
	Letting $\epsilon \to 0$ implies that case (1.b) occurs with probability zero.
	
	\medskip
	
	\noindent \emph{\textbf{Case 2}}. $\Theta > p$. We define $M$ as a stopping time with the following property: $M$ is the first time where $M > \tau_k^A $ and $M > \tau_k^B  $ and Bob plays left at time $M$ under $S_B$. If there is no such time, then $M = \infty$. 
	Let $S_B'$ be the following Bob strategy: play $S_B$ until time $M-1$. From round $M$ on, play right.
	
	Similarly to the neutral players case, let $\widetilde{H}_{M}$ be Alice's public history (i.e. containing the sequence of arms she played but not her rewards) running from round $0$ until the end of round $M-1$. If $M = \infty$, then $\widetilde{H}_M$ is the whole history. Note that $\widetilde{H}_M$ is observable by both players at the beginning of round $M+1$. For $M < \infty$, define
	$$
	\widetilde{H}_M \mbox{ is {good} } \iff \mu_{\widetilde{H}_M}\bigl(p,p+ 2\delta\bigr) < \sqrt{\epsilon} \; \mbox{and} \;
	\mu_{\widetilde{H}_M}\bigl(0,p\bigr) < \epsilon \delta\,.
	$$
	Let $D = \{(\widetilde{H}_M \mbox{ is good}) \cap (M < \infty)\}$.
	
	\medskip
	
	By combining cases $(1.a)$ and $(1.b)$ for competing players, we get that
	$$
	\Pr\Bigl( (\Theta < p) \cap (R_{\infty}^A = \infty)) = 0
	$$
	Thus, for sufficiently large $k$, we get
	\begin{align} \label{eq:1a1b_competing_combo_ub}
		\Pr\Bigl( (\Theta < p) \cap (R_{\infty}^A \geq k+1)) < \epsilon^2 \delta\,.
	\end{align}
	We follow the formulas from (\ref{eq:posterior_bob}) to (\ref{eq:bad_history_finite_Q_neutral}) of the analysis for neutral players, with the only change that in (\ref{eq:exp_bound_barmu}) we use the bound from (\ref{eq:1a1b_competing_combo_ub}). Thus we obtain  the same bound as in (\ref{eq:bad_history_finite_Q_neutral}):
	\begin{align}  \label{Q:bound_bad_competing}
		\Pr\bigl((M < \infty) \cap (\widetilde{H}_M \mbox{ is bad})\bigr)
		\leq 2 \sqrt{\epsilon}\,.
	\end{align}
	Define $S_B'$ as the following Bob strategy: play $S_B$ until the end of round $M-1$. At round $M$, if $D$ holds, then play right forever. Otherwise, continue with $S_B$.
	On the event $D$, since $\mu_{\widetilde{H}_M}(0,p) < \epsilon \delta$, by playing right forever from round $M+1$ onwards Bob guarantees a minimal utility of at least
	$$
	\Ex\Bigl[\bigl(u_{M+1}^B(S_A, S_B') \bigr) \mathbbm{1}_{D} \mid \widetilde{H}_M \Bigr] \geq - \frac{\epsilon \delta}{1-\beta} \cdot \mathbbm{1}_D
	$$
	
	Let $\gamma_i(M)$ be the reward of player $i$ in round $M$ under strategies $(S_A, S_B)$ and $\gamma_i'(M)$ the reward of player $i$ in round $M$ under strategies $(S_A, S_B')$. Note that $\gamma_A(M) = \gamma_A'(M)$.
	Then
	\begin{small}
		\begin{align}
			\Ex\Bigl[\bigl(u_{M}^B(S_A, S_B') \bigr) \mathbbm{1}_{D} \mid \widetilde{H}_M \Bigr]
			& =
			\Ex\Bigl[\bigl(\gamma_B'(M) - \gamma_A'(M) \bigr) \mathbbm{1}_{D} \mid \widetilde{H}_M \Bigr] + \Ex\Bigl[\bigl(u_{M+1}^B(S_A, S_B') \bigr) \mathbbm{1}_{D} \mid \widetilde{H}_M \Bigr] \notag  \\
			& \geq \left(p + \frac{\delta}{2}\right) \mathbbm{1}_{D} - \Ex\Bigl[\bigl(\gamma_A(M) \bigr) \mathbbm{1}_{D} \mid \widetilde{H}_M \Bigr] - \frac{\epsilon \delta}{1-\beta} \cdot \mathbbm{1}_D
		\end{align}
	\end{small}
	For strategy pair $(S_A, S_B)$ we get
	\begin{small}
		\begin{align}
			\Ex\Bigl[\bigl(u_{M}^B(S_A, S_B) \bigr) \mathbbm{1}_{D} \mid \widetilde{H}_M \Bigr]
			& =
			\Ex\Bigl[\bigl(\gamma_B(M) - \gamma_A(M) \bigr) \mathbbm{1}_{D} \mid \widetilde{H}_M \Bigr] + \Ex\Bigl[\bigl(u_{M+1}^B(S_A, S_B) \bigr) \mathbbm{1}_{D} \mid \widetilde{H}_M \Bigr] \notag  \\
			& = p \cdot  \mathbbm{1}_D - \Ex\Bigl[\bigl(\gamma_A(M) \bigr) \mathbbm{1}_{D} \mid \widetilde{H}_M \Bigr] + \Ex\Bigl[\bigl(u_{M+1}^B(S_A, S_B) \bigr) \mathbbm{1}_{D} \mid \widetilde{H}_M \Bigr]
		\end{align}
	\end{small}
	Since $\Ex\Bigl[\bigl(u_{M}^B(S_A, S_B) \bigr) \mathbbm{1}_{D} \mid \widetilde{H}_M \Bigr]  \geq \Ex\Bigl[\bigl(u_{M}^B(S_A, S_B') \bigr) \mathbbm{1}_{D} \mid \widetilde{H}_M \Bigr] $, by
	combining the previous inequalities we obtain
	\begin{align}
		& p \cdot  \mathbbm{1}_D - \Ex\bigl[\bigl(\gamma_A(M) \bigr) \mathbbm{1}_{D} \mid \widetilde{H}_M \bigr] + \Ex\bigl[\bigl(u_{M+1}^B(S_A, S_B) \bigr) \mathbbm{1}_{D} \mid \widetilde{H}_M \bigr] \geq \notag \\
		& \left(p + \frac{\delta}{2}\right) \mathbbm{1}_D - \Ex\Bigl[\bigl(\gamma_A(M) \bigr) \mathbbm{1}_{D} \mid \widetilde{H}_M \Bigr] - \epsilon \delta \implies \notag \\
		& \Ex\bigl[\bigl(u_{M+1}^B(S_A, S_B) \bigr) \mathbbm{1}_{D} \mid \widetilde{H}_M \bigr] \geq \frac{\delta}{2} \cdot  \mathbbm{1}_D
	\end{align}
	
	On the other hand, Alice can ensure a utility of at least $-\epsilon \delta/(1-\beta)$ from round $M+1$ onwards by always playing right, and so
	\begin{align}
		\Ex\bigl[\bigl(u_{M+1}^A(S_A, S_B) \bigr) \mathbbm{1}_{D} \mid \widetilde{H}_M \bigr] \geq - \frac{\epsilon \delta}{1-\beta} \cdot  \mathbbm{1}_D
	\end{align}
	
	Since the game is zero-sum, we obtain that
	\begin{align}
		0 = \Ex\bigl[u_{M+1}^B(S_A, S_B) \bigr] + \Ex\bigl[u_{M+1}^A(S_A, S_B) \bigr] \geq \Bigl( \frac{\delta}{2} - \frac{\epsilon \delta}{1-\beta} \Bigr) \cdot \Pr(D)
	\end{align}
	By choice of $\epsilon$ in (\ref{eq:epsilon_bound_competing}) we have $\delta / 2 - \epsilon \delta/(1-\beta) > 0$, so $\Pr(D) = 0$.
	Combined with inequality (\ref{Q:bound_bad_competing}), we get that $\Pr(M < \infty) < 2 \sqrt{\epsilon}$.
	This implies
	$$\Pr\Bigl((R_{\infty}^A = R_{\infty}^B = \infty) \cap \mbox{\{Bob plays left infinitely often\}}\Bigr) < 2 \sqrt{\epsilon}
	$$
	By letting $\epsilon \to 0$, we conclude the latter probability must be zero.
	By symmetry, we obtain that if Bob and Alice explore infinitely often, then they both settle on the right arm from some point on with probability $1$.
\end{proof}

\section{Competitive Play -- Improved Bounds for a Uniform Prior} \label{app:competitive}

In this section we give improved bounds for the thresholds in the case of an arm with a uniform prior, for the case where the players are competing ($\lambda = -1$). In particular, we show that both players will explore the risky arm for all $p < 5/9$ and will not explore for all $p > 2 - \sqrt{2}$.

\begin{proposition} \label{prop:5_9}
	Let arm $L$ have a known probability $p$ and arm $R$ with a uniform prior that is common knowledge. Then both players will explore arm $R$ with positive probability for all $p < 5/9$ as the discount factor $\beta \to 1$.
\end{proposition}
\begin{proof}
	To prove this, we consider the scenario where in round zero Bob is forced to play $L$ and Alice is forced to play $R$, while the players can play optimally afterwards. Then we show that Alice wins for all $p < 5/9$ as $\beta \to 1$ regardless of Bob's strategy, which will imply that in any equilibrium the players will both play the right arm in round zero.
	
	Consider now the scenario where they are forced to play as described above and define the following strategy for Alice:
	\begin{enumerate}
		\item If the bit observed in round zero is 0, then play L in round 1. Then
		\begin{itemize}
			\item If Bob played R in round 1, then play L in round 2 and from round 3 onwards copy Bob's arm from the previous round.
			\item If Bob played L in round 1, then play optimally from round 2 onwards.
		\end{itemize}
		\item Else, if the bit observed in round zero is 1, then play R again in round 1. Play optimally from round 2 onwards.
	\end{enumerate}
	We show that Alice wins in expectation regardless of Bob's counterstrategy. For round zero we have $\Ex(\gamma_A(0)) = 1/2$ and $\Ex(\gamma_B(0)) = p$.
	Alice's expected payoff in round 1 depends on whether she got a 0 or a 1 in round zero: $$\Ex(\gamma_A(1)) = 1/2 \cdot p + 1/2 \cdot 2/3 =  1/2 \cdot p + 1/3\,.$$
	
	To analyze Alice's payoff in round two and Bob's maximum expected reward overall we consider two cases, depending on Bob's move in round one. We assume that Bob plays optimally from round two onwards. Note also that Bob's decision for what arm to play in round one cannot depend on Alice's bit, so it is in fact determined at round zero.
	
	\medskip
	
	\noindent \emph{Case 1}: Bob plays R in round one. Denote by $\gamma_A(t | 1, R)$ Alice's reward in round $t$ given that she saw a 1 in round zero and that Bob's strategy is to play R in round one, and similarly for $\gamma_A(t | 0, R)$. Then Alice's expected total reward is:
	\begin{small}
		\begin{align} \label{eq:ev_Alice}
			\Ex(\Gamma_A) & = \Ex(\gamma_A(0)) + \Ex(\gamma_A(1)) \cdot \beta + \Ex(\gamma_A(2)) \cdot \beta^2 + \sum_{t = 3}^{\infty} \Ex(\gamma_A(t)) \cdot \beta^t =  \frac{1}{2} + \beta\left(\frac{p}{2} + \frac{1}{3} \right) \notag \\
			& \; \; \;  + \beta^2 \left(\frac{p}{2} + \frac{\Ex(\gamma_A(2|1,R))}{2}   \right)+ \frac{1}{2}  \sum_{t = 3}^{\infty} \Ex(\gamma_A(t | 1,R)) \cdot \beta^{t} + \frac{1}{2}  \sum_{t = 3}^{\infty} \Ex(\gamma_B(t-1 | 0, R)) \cdot \beta^{t} \,.
		\end{align}
	\end{small}
	Bob's expected total reward is:
	\begin{align} \label{eq:ev_Bob_rightcase}
		\Ex(\Gamma_B) = p + \frac{1}{2} \cdot \beta + \frac{1}{2}  \sum_{t= 2}^{\infty} \Ex(\gamma_B(t | 0, R)) \cdot \beta^t + \frac{1}{2}\sum_{t= 2}^{\infty} \Ex(\gamma_B(t | 1, R)) \cdot \beta^t
	\end{align}
	
	In the case where Alice observes a 1 in round zero, then Alice has an advantage from round two onwards, so $\Ex(\gamma_A(t | 1,R)) \geq \Ex(\gamma_B(t | 1,R))$ for all $ t \geq 2$.
	Using this fact and simplifying the expressions, we obtain that
	Alice's net gain in the zero sum game is
	\begin{align}
		\Ex(\Gamma_A) - \Ex(\Gamma_B) 
		& \geq \frac{1}{2} + \beta \left(\frac{p}{2}  + \frac{1}{3} \right) + \frac{p \beta^2}{2} - p - \frac{\beta}{2} - \frac{1 - \beta}{2} \cdot \sum_{t = 2}^{\infty} \Ex(\gamma_B(t | 0, R)) \cdot \beta^t \notag
	\end{align}
	
	To bound $\Ex(\gamma_B(t | 0, R))$ we consider two scenarios, depending on whether Bob saw a zero or a one in round 1. Moreover, by round two Bob knows the bit seen by Alice in round zero since her behavior is different depending on that bit. Let $X_1$ and $X_2$ be two  random variables with densities $3(1-x)^2$ and $6x(1-x)$ for all $x \in [0,1]$ respectively. Then
	\begin{align} \label{eq:vbbound}
		\Ex(\gamma_B(t | 0, R)) & \leq \frac{2}{3} \cdot \Ex(\max(p, X_1)) + \frac{1}{3} \cdot \Ex(\max(p, X_2)) = 1/3 - 1/3 \cdot p^3\notag
	\end{align}
	
	Then as $\beta \to 1$, the difference in rewards is
	$$\lim_{\beta \to 1} \left[ \Ex(\Gamma_A) - \Ex(\Gamma_B) \right] \geq 1/6 + p^3 / 6 - p^2/2\,.$$
	Then Alice wins in expectation for all $p \leq 0.65$.
	
	\medskip
	
	\noindent \emph{Case 2}: Bob plays L in round one. Note we will write again $\gamma_A(t|\emph{b}, L)$ to denote Alice's reward in round $t$ given that Alice saw the bit $b$ in round zero and Bob plays left in round one. Bob's expected total reward is
	$\Ex(\Gamma_B)) = p + p \cdot \beta + \sum_{t= 2}^{\infty} \Ex(\gamma_B(t)) \cdot \beta^t$,
	while Alice's is 
	\begin{align}
		\Ex(\Gamma_A) & = \frac{1}{2} + \beta \left(\frac{p}{2} + \frac{1}{3} \right) + \frac{1}{2} \cdot \sum_{t = 2}^{\infty} \Ex(\gamma_A(t | 1,L)) \cdot \beta^{t} + \frac{1}{2} \cdot \sum_{t = 2}^{\infty} \Ex(\gamma_B(t-1 | 0, L)) \cdot \beta^{t}. \notag
	\end{align}
	
	Again Alice has an informational advantage so she can at least equalize from round two onwards regardless of the value of the bit observed in round zero, and so $ \Ex(\gamma_A(t | b, L)) \geq  \Ex(\gamma_B(t | b, L))$ for all $ t \geq 2$ and $b \in \{0,1\}$.
	The difference in expected payoffs can be bounded in this case by $$ \Ex(\Gamma_A) - \Ex(\Gamma_B) \geq 1/2  + \beta \left(p/2 + 1/3  \right) - p (1 + \beta )\,.$$
	Taking $\beta \to 1$ gives
	$$\lim_{\beta \to 1} \left[ \Ex(\Gamma_A) - \Ex(\Gamma_B) \right] \geq 5/6 - 3p/2\,.$$
	Then Alice wins for all $p < 5/9$ when Bob plays left.
	
	\medskip
	
	Taking into account both cases implies the players explore for all $p < 5/9$ as required.
\end{proof}

\begin{proposition} \label{prop:upper_bound_uniform}
	Let arm $L$ have a known probability $p > 2 - \sqrt{2}$ and arm $R$ with a uniform prior, both of which are common knowledge. Then with probability $1$ the players will not explore arm $R$ in any equilibrium.
\end{proposition}
\begin{proof}
	To analyze this, we study the game where round zero is fixed such that Alice is at the right arm and Bob is at the left arm.
	
	Consider the following strategy for Bob: stay at arm $L$ in rounds zero and one, then in each round $t$, where $t \geq 2$, copy Alice's move from round $t-1$. Suppose Alice plays optimally given Bob's strategy.
	Then Alice's expected total is:
	$$
	\Ex(\Gamma_A) = 1/2 + \sum_{t=1}^{\infty} \Ex(\gamma_A(t)) \cdot \beta^t\,.
	$$
	For Bob, note that $\Ex(\gamma_B(t)) = \Ex(\gamma_A(t-1))$ for all $t \geq 2$.
	Then his expected total reward is
	\begin{align}
		\Ex(\Gamma_B) & = p \cdot (1 + \beta) + \sum_{t =2}^{\infty} \Ex(\gamma_A(t-1)) \cdot \beta^t
		= p \cdot (1 + \beta) + \sum_{t =1}^{\infty} \Ex(\gamma_A(t)) \cdot \beta^{t+1} \notag
	\end{align}
	Bob's expected net gain is
	$$\Ex(\Gamma_B) - \Ex(\Gamma_A) = p \cdot (1 + \beta) - 1/2 - (1- \beta) \cdot \sum_{t=1}^{\infty} \Ex(\gamma_A(t)) \cdot \beta^t\,.$$
	Let $X$ be a random variable with a uniform prior.
	Then Alice's expected value at round $t$ can be bounded by the expected maximum of $p$ and $X$, so
	\begin{align} 
		\Ex(\gamma_A(t)) \leq \Ex(\max\left(p, X\right)) = \int_{0}^{p} p \mathop{dx} + \int_{p}^{1} x \mathop{dx} = \frac{1}{2} + \frac{p^2}{2}
	\end{align}
	Using the bound on $\Ex(\gamma_A(t))$ above gives
	$$
	\Ex(\Gamma_B) - \Ex(\Gamma_A) \geq p \cdot (1 + \beta) - \frac{1}{2} - (1- \beta) \cdot \sum_{t=1}^{\infty}  \left(\frac{1}{2} + \frac{p^2}{2}\right) \cdot \beta^t
	$$
	Taking $\beta \to 1$ gives $$\lim_{\beta \to 1} \left[\Ex(\Gamma_B) - \Ex(\Gamma_A)\right] \geq 2p - 1/2 -  \left(1/2 + 1/2 \cdot p^2\right) = 2p - 1 - 1/2 \cdot p^2 \,.$$
	Then for all $p > 2 - \sqrt{2}$ Bob's net gain is strictly positive, thus Alice is at a disadvantage by playing the right arm in round zero. It follows that Alice is better off by playing the left arm in round zero, and so the players will never explore arm $R$ in any equilibrium.
\end{proof}

\begin{figure}[H]
	\centering
	\includegraphics[scale=0.64]{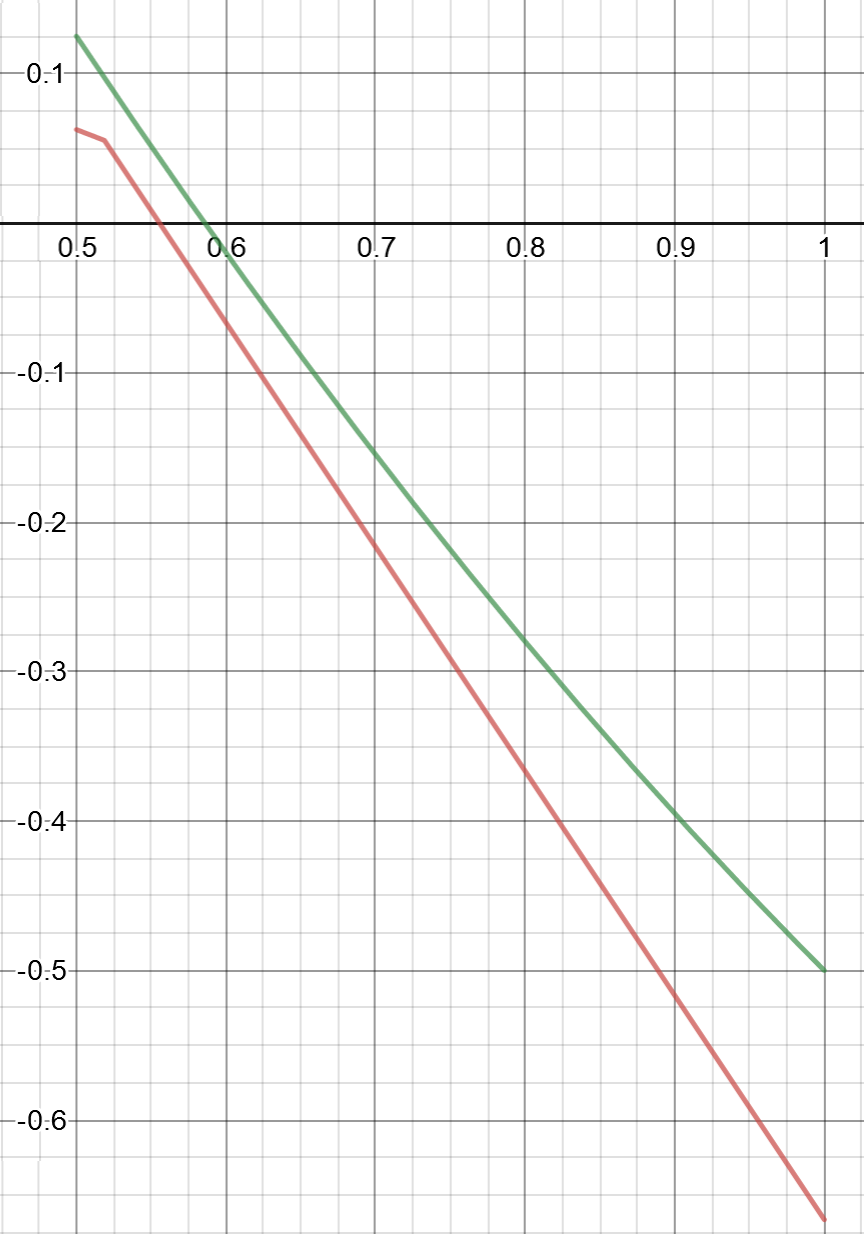}
	\caption{Bounds as a function of $p \in [0.5, 1]$ where the right arm has a uniform prior and $\beta \to 1$: the red line shows the lower bound on Alice's net gain given by the function $\mbox{lb}(p) = \min\{1/6 + p^3/6 - p^2/2, 5/6 - 3p/2\}$ (Proposition~\ref{prop:5_9}) when in round zero she starts at the right arm and Bob starts at left, after which they both play optimally. The green line shows the corresponding upper bound on Alice's net gain given by $\mbox{ub}(p) =p^2/2 + 1 - 2p$ (Proposition~\ref{prop:upper_bound_uniform}).}
	\label{fig:uniform_bounds}
\end{figure}

\section{Concluding Remarks and Questions} \label{sec:future_directions}

Several open questions arise from this work, such as understanding whether randomization is required sometimes, whether $p^* = \widetilde{p}$, and whether in the setting with multiple risky arms neutral and competing players settle eventually on the same arm with probability $1$ in every Nash equilibrium. More precisely, we propose the following questions:
\begin{enumerate}
	\item Do competing and neutral players always have optimal pure strategies or is randomization required sometimes?
	\item Recall the threshold $p^* = p^*(\mu, \beta, \lambda)$, defined in  Theorem~\ref{thm:competing_explore_less_discounted}, is the smallest number so that for all $p > p^*$ and in all Nash equilibria, the players do not explore.
	Similarly, $\widetilde{p} = \widetilde{p}(\mu, \beta, \lambda)$ from Theorem~\ref{thm:non_myopic_discounted}, is the largest number so that for all $p < \widetilde{p}$, the right arm is explored in some Nash equilibrium.
	\begin{enumerate}[(a)]
		\item Observation~\ref{obs:neutral_play_left_at_g} shows that for neutral players we have $\widetilde{p}(\mu, \beta, 0) \geq g(\mu, \beta)$ and Theorem~\ref{thm:cooperation} implies that $\widetilde{p}(\mu, \beta, 1) > g(\mu, \beta)$.
		Is $\widetilde{p}(\mu, \beta, 0) > g(\mu, \beta)$?
		\item Is $p^* = \widetilde{p}$ true in general?
		\item Are $p^*$ and $\widetilde{p}$ monotone in $\beta$ and $\lambda$?
	\end{enumerate}
	\item Consider the following variant of the model in this paper: each player learns the other player's rewards, but with a delay of $k$ rounds. Can one describe explicitly Nash equilibria in this variant as \cite{BH99} do in the model with perfect monitoring?
	\item Can one induce more exploration in this model by incorporating patent protection? For example, if Alice explores an arm for the first time, then a patent for her could mean that Bob cannot pull that arm for $k$ rounds afterwards, or alternatively, Bob should give Alice a fraction $\alpha$ of his reward whenever he pulls that arm in the $k$ rounds following her first exploration.
	\item When there are multiple risky arms, do neutral and competing players eventually settle with probability $1$ on the same arm in every Nash equilibrium? For neutral players, this is a conjecture by \cite{rotschild} made explicit by \cite{aoyagi}, who solved the case of two arms with finitely supported priors.
	\item Consider the setting with multiple risky arms.
	If there exist optimal pure strategies, can they be obtained from an index analogous to the Gittins index for a single player?
\end{enumerate}





\bibliographystyle{emss}
\bibliography{multiplayer_bib}


\end{document}